\documentclass[10pt, conference, letterpaper]{IEEEtran}

\usepackage[cmex10]{amsmath}
\usepackage{graphicx}

\usepackage[utf8]{inputenc}
\usepackage{amsthm}
\usepackage{amsfonts}
\usepackage[T1]{fontenc}  
\usepackage{enumerate}

\usepackage[ruled]{algorithm}
\usepackage[noend]{algpseudocode}

\newtheorem{theorem}{Theorem}
\newtheorem{corollary}[theorem]{Corollary}
\newtheorem{lemma}[theorem]{Lemma}
\newtheorem{fact}[theorem]{Fact}

\newtheorem{definition}[theorem]{Definition}

\newtheorem{proposition}[theorem]{Proposition}

\newif\iffull
\fullfalse

\usepackage{color}

\newcommand{\tj}[1]{#1} 

\newcommand{\labell}[1]{\label{#1}
   }

\newcommand{\lleft}{\text{left}}
\newcommand{\rright}{\text{right}}
\newcommand{\Lleft}{\text{Left}}
\newcommand{\Rright}{\text{Right}}
\newcommand{\dir}{\text{dir}}
\newcommand{\SingleRound}{{\sc SingleRound}}
\newcommand{\ReversedRound}{{\sc ReversedRound}}

\newcommand{\m}[1]{\mathcal{#1}}
\newcommand{\prob}{\text{Prob}}
\newcommand{\NAT}{{\mathbb N}}
\newcommand{\comment}[1]{}
\newcommand{\modd}{\text{mod }}
\newcommand{\ID}{\text{ID}}
\newcommand{\RI}{\text{RI}}

\begin{document}

\title{Deterministic Symmetry Breaking in Ring Networks}

\author{
\IEEEauthorblockN{
Leszek G\k{a}sieniec\IEEEauthorrefmark{1},
Tomasz Jurdzinski\IEEEauthorrefmark{2},
Russell Martin\IEEEauthorrefmark{1} and
Grzegorz Stachowiak\IEEEauthorrefmark{2} 
}
\IEEEauthorblockA{%
\IEEEauthorrefmark{1}%
Department of Computer Science, The University of Liverpool, United Kingdom}
\IEEEauthorblockA{%
\IEEEauthorrefmark{2}%
Institute of Computer Science, University of Wroc{\l}aw, Wroc{\l}aw, Poland}
}

\maketitle

\begin{abstract}
We study a distributed coordination
mechanism for uniform agents located on a circle.
The agents perform their actions in synchronised rounds. At the beginning
of each round an agent chooses the direction of its movement from
clockwise, anticlockwise, or idle, and moves at unit speed during this round.
Agents are not allowed to overpass, i.e., when an agent collides with
another it instantly starts moving with the same speed in the opposite
direction (without exchanging any information with the other agent). However, 
at the end of each round
each agent has access to limited information regarding
its trajectory of movement during this round. 
%
We assume that $n$ mobile agents are initially located on a circle unit circumference
at arbitrary but distinct positions unknown to other agents.
The agents are equipped with unique identifiers from a fixed range.
The {\em location discovery} task to be performed by each agent is to determine
the initial position of every other agent.

Our main result states that, if the only available information
about movement in a round is limited to 
distance between the initial and the final position, then there
is a superlinear lower bound on time needed to solve the location
discovery problem. Interestingly, this result corresponds
to a combinatorial symmetry breaking problem, which might
be of independent interest. If, on the other hand, an agent has
access to the distance to its first collision with another agent
in a round, we design an asymptotically efficient and close to
optimal solution for the location discovery problem.
%
\iffull
Assuming that agents are anonymous (there are no IDs distinguishing them),
our solution applied to randomly chosen IDs from appropriately 
chosen range
gives an (almost) optimal algorithm, 
improving upon the complexity of previous randomized
results. 
\fi
\end{abstract}

\begin{IEEEkeywords}
mobile robots, location discovery, bouncing
\end{IEEEkeywords}

\section{Introduction}\labell{sec:intro}
One of the most studied network topologies in the context
of distributed computation, as well as coordination
mechanisms for mobile agents, is the ring network \cite{AW-book,kkm-book,lynch-book}.
Recently, studies of geometric ring networks were initiated in the context of terrain exploration by 
agents/robots with limited communication and navigation capabilities~\cite{CzyzowiczGKKPP12,FriedetzkyGGM12}.
This refers to the concept of swarms, i.e., large groups of limited but
cost-effective entities (robots, agents) that can be deployed to perform an exploration
in a hard-to-access hostile environment. 
The usual swarm robot properties include anonymity, negligible dimensions, no explicit communication, and no common coordinate system (cf.\ \cite{SuzukiY99}). 
Some of these models 
assume limited visibility of the surrounding environment and asynchronous
operation. In most situations involving such weak robots, the fundamental research question
concerns the feasibility of solving a given task (cf.\ \cite{DasFSY10,FlocchiniPSW08}). 
The cost of the algorithm is usually measured in terms of length of
a robot's walk or the time needed to complete the task.
There are several algorithmic solutions providing efficient distributed coordination
mechanisms in a variety of models, e.g.~\cite{CieliebakFPS12,Susca07,SuzukiY99}.
The dynamics of ``beads on a ring'' and billiard systems is also of independent 
interest, e.g.~\cite{Cooley05}.

One of the fundamental tasks in ad hoc distributed environments is to determine the
actual network topology. 
This topic was studied in networks modeled as graphs~\cite{BenderS94,ChalopinFMS10,FraigniaudGKP06}, as well as
networks deployed in a geometric environment~\cite{ArkinFM00,CzyzowiczLP11,DengKP98,HoffmannIKK01}.
Most of those solutions work under the assumption that neighbors (in a graph)
can exchange messages,
or that agents have some visibility allowing them to inspect their
nearby neighborhood.

In the case of networks containing swarm robots, communication
and visibility capabilities are often severely restricted. Lack of these capabilities
in some settings can be overcome by the possibility of agents monitoring their
own trajectories, sensing collisions with other agents, or inferring some information
from the fact that all agents behave in a fixed regular fashion.
Another factor simplifying various tasks might be a restriction on the class
of environments or the allowed movement trajectories of agents.

Following \cite{FriedetzkyGGM12,CzyzowiczGKKPP12} 
we consider a model where 
the agents operate in synchronised rounds, 
and they lack direct means of communication.  
The trajectory of an agent 
in a given round is
represented as a continuous curve that connects the start and the end
points of the route adopted by the agent. While moving along their trajectories
the agents collide with their immediate neighbours, and information on the exact
location of those collisions might be recorded and further processed. 
%
When agents are located on a circle, each agent
may eventually conclude on the relative location of all agents' initial positions,
even given only limited information about its trajectory, e.g., at specified time intervals.
This, in turn, enables other distributed mechanisms based on
full synchronisation, e.g.\ equidistant distribution along the circumference of the
circle and an optimal boundary patrolling scheme. Most of the models adopted in the literature
on swarms assume that the agents are either almost or entirely oblivious,
i.e., throughout the computation process the agents follow a very simple, rarely
amendable, routine of actions. 
Such a scenario is studied in 
\cite{CzyzowiczGKKPP12,CzyzowiczKP13,CzyzowiczDKP14}, where agents are entirely
oblivious but can register all their collisions. (In \cite{CzyzowiczKP13,CzyzowiczDKP14} 
agents might have different velocities, and in \cite{CzyzowiczDKP14} they might have
different masses.)  
In this paper we adopt the model from \cite{FriedetzkyGGM12}, where even
the possibility of an agent tracking its own trajectory is severely limited.
(The model we study can also be seen as a variation of
that studied in \cite{ASY}.)
In order to overcome this weakness, more adaptivity of behavior is allowed.
So, the ultimate goal of this line of research is to determine
how much information about their trajectories agents need to 
solve some communication or exploration problems, and how
efficiently these problems can be solved.

Our focus is on deterministic solutions for these communication and 
exploration problems for agents having unique IDs, which is necessary for
symmetry breaking.  However, our results can be applied to randomly 
chosen IDs from an appropriately chosen range to improve upon the complexity of 
previous randomized results.  Due to space reasons, those adaptations
will not be discussed in this paper.  

\subsection{Model}
\newcommand{\basic}{basic}
\newcommand{\lazy}{lazy}
\newcommand{\perceptive}{perceptive}
\newcommand{\pos}{\text{dist}} 
\newcommand{\coll}{\text{coll}} 

A network $A$ is deployed on a circle with
circumference one, along which $n$ agents (i.e., the elements of $A$)
move and interact in 
synchronised rounds, where each round lasts one unit of time. 
%
%
The agents do not necessarily share the same sense of
direction, i.e., while each agent distinguishes between its own clockwise (C) and
anticlockwise (A) directions, agents may not have a coherent view on this. 
The direction ``clockwise'' is also called ``{\rright}'', and we also refer
to ``anticklockwise'' as ``{\lleft}''.
At the beginning of a round, an agent $a$ 
assigns one of the values from the set
$\{\text{idle}, \rright, \lleft\}$ to its local
variable $\dir_a$. When the option ``idle''
is chosen, the agent starts the round without moving in any
direction.
In the case that $\dir_a=\rright$ or $\dir_a=\lleft$, 
the agent starts the round
moving at unit speed on the circle in the direction $\dir_a$. We assume that
agents are not allowed to overpass each other along the circle. When two agents 
moving in the opposite directions collide with each other, they instantly start moving with the
same speed but in the opposite directions.
If an agent $a$ moving in the direction $\dir\in\{\rright,\lleft\}$ collides
with another agent $a'$ which is currently idle, then $a$ stays idle after the collision 
and $a'$ immediately starts moving in the direction $\dir$ (i.e., in the same ``objective'' direction 
in which $a$ was moving before the collision, irrespective of the fact whether 
$a$ and $a'$ have consistent senses of direction).
The agents cannot leave marks on
the ring, they have zero visibility, and they cannot exchange messages. Instead,
during each round each agent has access to some (specified) information about
its trajectory during this round. This information can be processed or stored
for further analysis. 
Since the agents never overpass, we may assume that the agents are arranged
in an implicit (i.e.\ never disclosed to the agents) periodic order from $a_1$ to $a_n$.

Each agent has access to its relative position at the end of a round; more precisely,
it knows the distance $\pos()$ to the right (according to its own sense of direction)
between its position at the beginning of the round
and the position at the end of the round, measured in the agent's
clockwise direction.  In other words, there is no ``universal'' coordinate system on the circle, 
the distance is measured relative to the starting position of an agent at the start of the round.  
%
We distinguish three variants of the model:
\begin{itemize}
\item
{\em {\basic}} -- an agent is \textbf{not} allowed to start a round idle, it has to start
moving either in the {\rright} or the {\lleft} direction;
\item
{\em {\lazy}} -- an agent is allowed to start a round idle, moving right or left;
\item
{\em {\perceptive}} (or $1$-\perceptive) -- this is the basic model with the additional feature
that an agent gets the value $\coll()$ at the end of each round, which is equal to the distance between its
position at the beginning of the round and the position of its first collision in that round.
\end{itemize}
Thus, the {\basic} model is the weakest one. The {\lazy} model extends the {\basic} model by
increasing an agent's freedom in choosing various movement options. 
The {\perceptive} model, on
the other hand, extends the {\basic} model by providing more information 
about an agent's own trajectory to itself.

\subsection{Notation and definitions}

In this paper we address deterministic algorithms which require (for symmetry breaking)
that agents have unique identifiers (IDs). We assume that each ID is a natural number in the
set $\{1,\ldots,N\}$ and each agent is aware of the value of $N$.
We also consider randomized algorithms, 
and in this case the agents are uniform
and anonymous. That is, they are indistinguishable from 
other agents; in particular, no IDs are provided in this case. 



The actual number of agents is denoted by $n$. In general, we assume that the
only information available to agents about $n$ is whether $n$ is odd or even.
Additionally, we assume that $N\ge n>4$.\iffull\footnote{At the end of the paper
we address the issue how to determine that $n\leq 4$ and how efficiently the parity of $n$ can be determined by the agents.}\fi

For an agent $a$, $\ID_a$ denotes the identifier of $a$, and $\ID_a[i]$ denotes
the $i$th bit of $\ID_a$. We also assume that at the beginning of each round, 
each agent $a$ can set a local
variable $\dir_a$ with value {\lleft}, {\rright} or idle (only in the lazy model), 
and the value $\dir_a$ (in general) determines the 
way in which $a$ starts moving in the next round.
For natural numbers $i$ and $j$, let $[i,j]=\{k\in\NAT\,|\, i\leq k\leq j\}$ and
let $[i]=[1,i]$.

By {\em right ring distance} between agents $a$ and $a'$ we mean $1$ plus 
the number of agents on the ring between $a$ and $a'$ going from $a$ to $a'$ in the clockwise direction.
The {\em left ring distance} is defined analogously.
If no common sense of direction is established, the right/left distance from the point of
view of an agent is measured according to its own sense of direction.
Observe that, by the model's restrictions,
the relative order of agents on the ring does not change. 
Thus, the ring distance between agents does not change
during executions of algorithms.
For an agent $a$, $N_a(k)$ denotes the set of agents in ring distance at 
most $k$ from $a$.

Let $\m{S}=(S_1,\ldots,S_k)$ be a sequence of subsets of $[N]$. We say that agents {\em execute}
$\m{S}$ in a sequence of $k$ rounds if the agent $a\in[N]$ sets $\dir_a={\rright}$ in the
$i$th round iff $a\in S_i$; otherwise $\dir_a=\lleft$.
Moreover, given a set $A'\subseteq A$ of ``marked'' agents we say that $\m{S}$ is {\em executed}
on $A'$ if agents from $A'$ set their directions in consecutive rounds according to $\m{S}$,
while each $a\in A\setminus A'$ sets $\dir_a$ to $\rright$ in each round.

\subsection{A basic tool}
Let an $(n_C,n_A)$-round be any round in which $n_C$ agents start the round clockwise
and $n_A$ agents start the round anticlockwise (according to some ``objective''
sense of direction).
A simple but key property of the ring networks was observed in \cite{FriedetzkyGGM12}.
\begin{lemma}\labell{lem:round}\cite{FriedetzkyGGM12}
Assume that the positions of agents $a_1,\ldots,a_n$ at the start of an $(n_C,n_A)$-round
are $p_1,\ldots,p_n$. Then,
during 
the round all agents are rotated along the initial positions by a rotation
index of $r = (n_C-n_A)\ \modd n$, i.e., the position of $a_i$ at the end of the round
is $p_{1+(i-1+r)\modd n}$. 
\end{lemma}
By the above lemma, each agent experiences the same shift by $r$ places in a round.
Therefore, we define the {\em rotation index} of a round as the number of places by which
agents move in that round in the clockwise direction. Thus, the rotation index of 
an $(n_C,n_A)$-round is equal to $(n_C-n_A)\ \modd n$.

In this paper, \SingleRound\ denotes one
round of computation in which each agent $a$ starts moving in the direction $\dir_a$.
\ReversedRound\ denotes one
round of computation in which each agent $a$ starts moving opposite to the direction $\dir_a$.
Note that, after an execution of {\SingleRound} followed by {\ReversedRound}, each agent $a$ gets
to the position occupied by $a$ before these two rounds transpired, provided agents do not change their
local variables $\dir_a$ in between the two rounds.

\subsection{Problems considered in the paper and previous results}

The main goal of this paper is to 
evaluate the feasibility and complexity of the {\em location discovery} (LD) problem 
in the models we consider.
The {\bf location discovery problem} is to determine
the initial position (i.e.\ starting position when all agents simultaneously ``wake up'' to begin the procedure) 
of every other agent\footnote{In \cite{FriedetzkyGGM12}, it is required that eventually each agent stops at
its initial position. In this paper this requirement
is ignored. A simple way to achieve this is to reverse all rounds of the
algorithm (see properties of \SingleRound\ and \ReversedRound). However, in our solutions agents 
collect information which allows them to get back on the initial positions much faster than by 
reversing all steps of an original algorithm.}.
That is, at the end of an execution of an algorithm, each agent $a\in A$ should know initial positions of 
all other agents, with respect to its own initial position.

We consider several problems which turn out to be efficient
tools for solving the location discovery problem. Moreover, they are interesting
as themselves, since they are useful in designing more complicated
communication mechanisms. 
Below, we define these problems.


\noindent{\bf Direction agreement.} The {\em direction agreement} is to agree on
which direction is {\em clockwise} and which is {\em counterclockwise}.
That is, at the end of the direction agreement procedure
all agents have coherent view on which direction
is clockwise, independent of any ``objective'' sense of direction.

\noindent{\bf Leader election.} The {\em leader election} problem 
is solved when {\em exactly} one agent
is assigned the status ``leader'' and all other agents have the status
``non-leader''. (Note that we do not require that non-leaders know the 
ID of the leader or any other information about it.)


\noindent{\bf Nontrivial move problem.}
We say that a round is a {\em trivial move} if its rotation index 
belongs to the set $\{0,n/2\}$ and it is a {\em nontrivial move} otherwise.
The {\em nontrivial move} problem is to assign to each agent $a$ its
direction $\dir_a$ such that if $a$ starts a round in the direction $\dir_a$,
then this round is a nontrivial move.

\smallskip

For the direction agreement, leader election, and the nontrivial move
problem we use the notion of {\em coordination problems}.

\smallskip

As a tool for solutions of other problems, we also consider the emptiness
testing problem.

\noindent{\bf Emptiness testing.}
Let $A\subseteq [N]$ denote the set of IDs of agents in the network.
{\em Emptiness testing} is a protocol which given $B\subseteq[N]$, determines
whether $B\cap A=\emptyset$. (That is, each agent $a\in A$ knows $B$ as an input
and it is aware of the
fact whether $A\cap B\neq\emptyset$ at the end of an execution of the protocol.)

The location discovery problem in the {\basic} and {\perceptive}
model were studied in~\cite{FriedetzkyGGM12}. 
It has been shown that there exists a randomized solution 
for anonymous networks (i.e.\ for identical
agents without IDs) working in time $O(n\log^2n)$ with high
probability in the 
{\perceptive} model. If $n$ is odd, this solution works
also under the assumptions of the {\basic} model.
In \cite{CzyzowiczGKKPP12}, oblivious algorithms are studied,
in which an agent is not allowed to change its direction at 
the beginning of a round. However, agents have access to
positions of all their collisions during a round.
It has been shown that, for some initial configurations,
the location discovery problem is infeasible in this model.
On the other hand, there is a family of initial configurations
for which the location discovery can be solved efficiently
in (sub)linear time.


\subsection{Our results}\labell{sec:res}
In this paper, we examine the complexity of deterministic leader election,
nontrivial move, direction agreement, and location discovery problems. 
We also study the impact
on the complexity of these problems 
of the parity of $n$, and whether agents initially share the same sense
of direction. In all considered settings we obtain results which are optimal
or close to optimal (see Tables \ref{tab:gen} and \ref{tab:common}). 

First, we show that the complexity of all coordination problems is asymptotically
equal up to an additive $O(\log N)$ factor. This gives an efficient and simple
solution for location discovery when $n$ is odd (Section~\ref{sec:red}).


%

The key technical contribution of the paper states that lack of the common 
sense of direction for even $n$ substantially changes 
the complexity of all considered problems, at least in
the {\basic} and {\lazy} model.
That is, the complexity of all coordination problems and position discovery is superlinear 
with respect to $n$ for $n=O(N^{1-\epsilon})$ and
constant $\epsilon>0$. 
More precisely, all considered  problems 
require $\Omega(n\log(N/n)/\log n)$ rounds in this setting (see Table~\ref{tab:gen}).
%
%
%
The reason for these large lower bounds is that the considered tasks 
require the solution of a kind of ``symmetry-breaking'' problem.
%
We define a purely combinatorial notion of a {\em distinguisher} 
(see Section~\ref{sec:basic}) to describe this symmetry-breaking problem 
which we think might be of independent interest.
Using the probabilistic method,
we also show that this bound is 
tight.

For the perceptive model, we provide a construction which
solves the nontrivial move problem in $O(\sqrt{n}\log N)$ rounds,
thus the
lower bound $\Omega(n\log(N/n)/\log n)$
does not hold for this case.

We also show that using solutions of the coordination 
problems considered in the paper, 
the location discovery problem
can be solved in $n+o(n)$ rounds in the lazy model (or basic model with odd $n$) and
in $n/2+o(n)$ rounds in the perceptive model, provided $\log N=o(\sqrt{n})$
(see the last columns of Tables~\ref{tab:gen} and \ref{tab:common} for details). 
These results are optimal up to additive $o(n)$ factors (using Lemma~\ref{l:disc:low} described
later).

\begin{table*}
\caption{Deterministic solutions in general setting} 
\label{tab:gen}
\begin{center}
\begin{tabular}{|c|c|c|c|c|}
\hline
& leader & nontrivial & direction & location\\
& election & move & agreement & discovery\\
\hline
\hline
odd $n$ & $O(\log N)$ & $\Theta(\log(N/n))$ & $O(1)$ & $n+O(\log N)$\\
\hline
basic model, even $n$ & $\Theta(\frac{n\log(N/n)}{\log n})$ & $\Theta(\frac{n\log(N/n)}{\log n})$ &$\Theta(\frac{n\log(N/n)}{\log n})$ & not solvable\\
\hline
lazy model, even $n$ & $\Theta(\frac{n\log(N/n)}{\log n})$ & $\Theta(\frac{n\log(N/n)}{\log n})$ &$\Theta(\frac{n\log(N/n)}{\log n})$ & $n+\Theta(\frac{n\log(N/n)}{\log n})$\\
\hline
perceptive model, even $n$ & $O(\sqrt{n}\log N)$ & ${O}(\sqrt{n}\log N)$ & ${O}(\sqrt{n}\log N)$ & $\frac{n}{2}+{O}(\sqrt{n}\log^2 N)$\\
\hline
\end{tabular}
\end{center}
\end{table*}

\begin{table*}
\caption{Deterministic solutions {\em with} common sense of direction}
\label{tab:common}
\begin{center}
\begin{tabular}{|c|c|c|c|c|}
\hline
& leader & nontrivial & location\\
& election & move & discovery\\
\hline
\hline
odd $n$ & $O(\log N)$ & $\Theta(\log(N/n))$ & $n+O(\log N)$\\
\hline
basic model, even $n$ & $O(\log^2N)$ & $O(\log^2 N)$  & not solvable\\
\hline
lazy model, even $n$ & $O(\log N)$ & $O(\log N)$  & $n+O(\log N))$\\
\hline
perceptive model, even $n$ & $O(\log N)$ & $O(\log N)$  & $\frac{n}{2}+{O}(\sqrt{n}\log N)$\\
\hline
\end{tabular}
\end{center}
\end{table*}

\iffull
Interestingly, our results can be applied to improve the best randomized
solutions in anonymous networks (without IDs of agents) \cite{FriedetzkyGGM12}
by assigning
IDs 
randomly from the range $[m^3]$, for an appropriately
chosen approximation, $m$, of the size $n$.
\else
Due to space limitations, proofs are omitted from this conference 
version.  \tj{They will be presented in the full version of the paper 
available on arXiv.} The Appendix contains proofs of the results 
in Section~\ref{sec:basic}
to give a flavor of the symmetry-breaking mechanism required for
the solution of these coordination problems.  
\fi

\subsection{Structure of the paper} 
%
First, in Section~\ref{sec:tools}, we provide some basic
facts and tools regarding the considered model which will
be used throughout the paper.
In Section~\ref{sec:red}, we establish relationships between asymptotic complexities
of coordination problems, summarized in Theorem~\ref{ABC}. We also discuss 
consequences of these reductions when the size $n$ of a network
is odd.

In Section~\ref{sec:basic}, the complexity
of 
the nontrivial move problem
in the basic model is examined. In particular, a superlinear
lower bound on the complexity 
of nontrivial move 
is shown, and an (almost) matching
upper bound is provided. 
\iffull
In Section~\ref{sec:lazy} we show that the lower bound on symmetry breaking from Section~\ref{sec:basic} applies in the {\lazy} model as well. 
\fi
In Section~\ref{sec:perc}, a construction allowing us to reduce
the complexity of loation discovery to $n/2+o(n)$ is described
in the {\perceptive} model. 

\iffull
As we assume that $n>4$ in (most of) this paper, and
often require that the parity of $n$ is known (e.g., to determine whether location
discovery is solvable or not), we complement the paper by discussing the problem 
about determining the parity of $n$, and solving the considered 
problems when $n\leq 4$. These issues are presented in Section~\ref{sec:special}. 
Then,
in Section~\ref{sec:random} we discuss how our solutions can be applied 
to build efficient randomized algorithms. 
\else
We assume that $n>4$ in (most of) this paper, and
often require that the parity of $n$ is known (e.g., to determine whether location
discovery is solvable or not).
The problem of determining the parity of $n$ will be discussed in the full version of
this paper, as will the case when $n \leq 4$.  Our solutions can be applied 
to build efficient randomized algorithms, but these issues are not discussed in
this version. 
\fi

\iffull
Finally, conclusions and open problems are
presented in Section~\ref{sec:open}.
\fi

\section{Basic Properties of the Model}\labell{sec:tools}
In this section we make a few observations regarding
features and limitations of the model studied in the paper.
\iffull
First, we provide tools allowing agents to infer some
knowledge from information about its traversed distances and observed
positions of collisions. Then, basic lower bounds on the 
complexity of location discovery are stated.
\fi
%

\begin{lemma}\labell{l:opp:ind}
All agents can determine in $O(1)$ rounds whether a rotation index of a given
round is $0$, $n/2$, larger than $n/2$ or smaller than $n/2$ (according to 
their own senses of directions).
\end{lemma}
\iffull
\begin{proof}
Assume that an algorithm runs two consecutive rounds with the same directions $\dir_a$
of all agents. Then, the sum $s$ of distances $\pos()$ on which a vertex v is shifted in these rounds
is larger than $1$ if and only if the rotation index of such a round is larger than $n/2$.
Similar relationships hold for other values of $s$ and rotation indexes.
\end{proof}
\else
\fi
For a fixed set of agents $A$, we define the {\em rotation index} $\RI(B)$ of a set $B$
as the rotation index of a round in which all elements of $B\cap A$ start the
round moving right (clockwise) and the remaining agents start the round moving left
(anticlockwise). 
(Note that we assume an objective sense of direction when
talking about agents which start a round moving clockwise/anticlockwise.)
Thus, $\RI(B)=(|B|-(n-|B|))\ \modd n=2|B|\ \modd n$. 
\iffull
\else
Below, we state some properties which can be proved using similar reasoning to that
in the proof of Lemma~\ref{l:opp:ind}.
\fi
\begin{lemma}\labell{l:ri}

\begin{enumerate}[(a)]
\item $\RI(B)=0$ if and only if $|B|\in\{0,n/2,n\}$.
\item If $\RI(B)\neq 0$, then $0<|B|<n$.
\item If $\RI(B)\neq 0$, and $B=B_1\cup B_2$ for disjoint $B_1,B_2$,
then $\RI(B_1)\neq 0$ or $\RI(B_2)\neq 0$.
\end{enumerate}
\end{lemma}
\iffull
\begin{proof}
Items (a) and (b) are obvious (see Lemma~\ref{lem:round}). Assume 
to the contrary that
$\RI(B)\neq 0$ and $\RI(B_1)=\RI(B_2)=0$ for a partition
$B_1,B_2$ of $B$. Then, (a) and (b) imply that
$|B_1|,|B_2|\in\{0,n/2,n\}$, $0<|B|<n$, and $|B|\neq n/2$. Moreover,
$|B_1|+|B_2|=|B|$. One can easily check that it is
impossible to satisfy all these relationships 
simultaneously.
\end{proof}
\fi


%

Now, we make an observation regarding information which can be inferred
by an agent using the distance between its starting position and 
the first collision in a round
(i.e., $\coll()$).
\begin{proposition}\label{prop:bounce}
Assume that an agent $b_0$ starts moving in a round in the direction $\dir_{b_0}$,
and let consecutive agents in the direction $\dir_{b_0}$ from $b_0$ be 
denoted $b_1,\ldots,b_{n-1}$.
Moreover, let the geometric
distance (on the ring) between $b_{i-1}$ and $b_i$ be $x_{i-1}$.
If $b_1,\ldots,b_k$ start the round in the direction $\dir_{b_0}$ for $k<n-1$,
and $b_{k+1}$ starts in the opposite direction to $\dir_{b_0}$, 
then the relative position of the
first collision of $b_0$ 
is equal to $(x_1+\cdots+x_{k})/2$.
\end{proposition}
\iffull
\begin{proof}
One can easily prove by induction on $j\geq 0$ the following fact:
$b_{k-j}$ collides with $b_{k-j+1}$ in distance
$(x_{k-j}+\cdots+x_k)/2$ from the initial position of $b_{k-j}$.
\end{proof}
\fi

\subsection{Lower bounds on the complexity of location discovery}
As observed by Friedetzky et al.~\cite{FriedetzkyGGM12}, location discovery
cannot be solved in the {\basic} model when $n$ is even. 
\begin{lemma}\labell{l:disc:imp}\cite{FriedetzkyGGM12}
It is impossible to solve the location discovery problem in the basic model
with even $n$.
\end{lemma}
The reason of this impossibility result follows from the fact that, when $n$ is 
even, the rotation
index of any round in the {\basic} model is always even. Therefore,
an agent can only visit positions of agents having even ring distance from itself.

\iffull
Below, we make an observation regarding lower bounds on the complexity
of location discovery.
\else
Below, we state the lower bounds on complexity of the location discovery problem.
Intuitively, they follow from the fact that each round gives one linear equation
with variables equal to distances between agents in the {\basic} and {\lazy} model,
while it provides two linear equations in the {\perceptive} model (as two distances
are given to an agent).
\fi
\begin{lemma}\labell{l:disc:low}
\begin{enumerate}
\item
The location discovery problem in the {\basic} and {\lazy} model cannot
be solved in less than $n-1$ rounds in the worst case.
\item
The location discovery problem in the {\perceptive} model cannot
be solved in less than $n/2$ rounds in the worst case.
\end{enumerate}
\end{lemma}
\iffull
\begin{proof}
The goal of each
agent is to determine real numbers $x_0,\ldots,x_{n-1}$, i.e.\ the 
distances between consecutive agents. At the beginning of the procedure, 
it is only known that $\sum_{i=0}^{n-1}x_i=1$.
%
The only information which
an agent gets in a round, is the sum $x_i+x_{(i+1)\modd n}+\cdots+x_{j\modd n}$ for some $1\leq i,j< n$ corresponding
to the relative distance between its starting and final positions in a round, 
i.e.\ the value of \pos() (see Lemma~\ref{lem:round}).
In the perceptive model, an agent gets additionally another linear 
combination, \coll(), of $x_1,\ldots,x_n$
equal to the distance to the first collision in a round.
This number is $1/2$ times the sum $x_i+x_{(i+1)\ \modd n}+\cdots+x_{j\ \modd n}$ 
for some $0\leq i,j< n$, by Proposition~\ref{prop:bounce}.
Therefore, the basic facts from linear algebra imply that the necessary
number of rounds to accomplish the task of position discovery is
$n-1$ in the {\basic} and {\lazy} models and $n/2$ in the perceptive model.
%
%
Indeed, otherwise an agent would have to determine $x_1,\ldots,x_n$ uniquely
from $m<n$ linear equations on variables $x_1,\ldots,x_n$.

\tj{NOT SURE IF THE FOLLOWING IS NECESSARY? CAN WE ARGUE SIMPLER ABOUT THAT?
One may argue that there is the following flaw in the above reasoning.
As agents can adapt their behavior on the basis of partial knowledge about
$x_0,\ldots,x_{n-1}$, the set of linear equations is not fixed in advance.
And, the agents can somehow encode information about their knowledge by
choosing their initial direction in consecutive rounds. 
However, 
%
let us consider (irrational) $x_0,\ldots,x_{n-1}$
such that $\sum_{i=0}^{n-1}x_i=1$ and the exact values of $x_0,\ldots,x_{n-1}$ 
%
cannot be determined uniquely from $n-2$ linear equations
on $x_1,\ldots,x_n$ of the type 
$x_i+x_{(i+1)\modd n}+\cdots+x_{j\modd n}=C$ for some $1\leq i,j< n$ 
and constant $C$.
(Assume that $(x_0,\ldots,x_{n-1})$ is obtained 
by choosing $x_i=1/n+\epsilon_i$ for $i<n-1$ and 
$x_{n-1}=1-\sum_{j<n-1}x_j$,
where $\epsilon_i$ is a random real number chosen with uniform distribution
from the interval $[-1/n^3,1/n^3]$. Then, with probability $1$, $x_0,\ldots,x_{n-1}$
satisfy the above conditions.)
Then, as long as the number of rounds is smaller than $n-1$ in the {\basic}
or {\lazy} model and smaller than $n/2$ in the {\perceptive} model,
each agent has insufficient information to determine all values from
$\{x_0,\ldots,x_{n-1}\}$.
}
\end{proof}
\fi

\section{Reductions between considered problems}\labell{sec:red}
In this section we establish reductions between the coordination problems.
The results 
\iffull
proved in this section 
\fi
are illustrated in 
Figures~\ref{fig:red:strong} and \ref{fig:red:weak}, and are summarized in 
Theorem~\ref{ABC}.
They work for arbitrary $n$,
provided $n>4$.

\begin{figure}[h]
\begin{center}
  \includegraphics[scale=0.85]{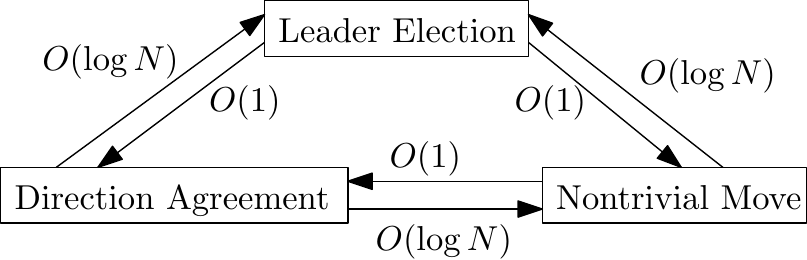}
  \caption{Complexity of reductions among coordination problems if 
	$n$ is odd or the model is either {\perceptive} or {\lazy}.}
\label{fig:red:strong}
\end{center}
\end{figure}

\begin{figure}[h]
\begin{center}
  \includegraphics[scale=0.85]{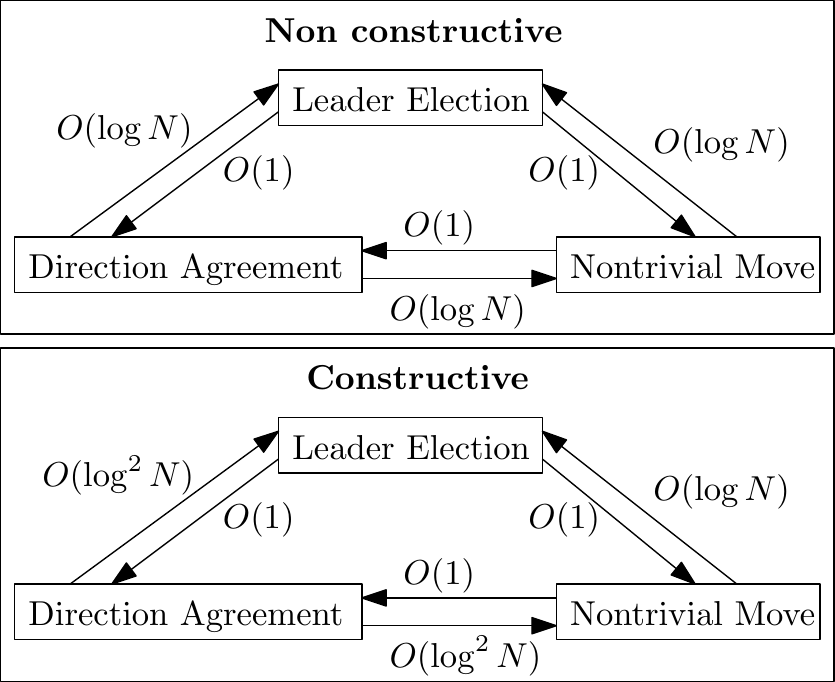}
  \caption{Complexity of reductions among coordination problems in the {\basic} model (even $n$).}
\label{fig:red:weak}
\end{center}
\end{figure}
\begin{theorem}\labell{ABC}
For each model considered in the paper (\basic, \lazy, \perceptive)
the asymptotic complexity of all coordination problems
(direction agreement, leader election, nontrivial move)
are equal up to an additive term $O(\log N)$.
\end{theorem}
\iffull
\else
\comment{\tj{Constructive variants of reductions 
are obtained by changing the sets of transmitters based on
values of consecutive bits of IDs (this approach faces
difficulty because of the fact that $(n/2,n/2)$-round and
$(0,n)$-round are indistinguishable in {\basic} and {\lazy} model). The nonconstructive variants apply the probabilistic
method.
}}
\fi

\subsection{The setting with the nontrivial move problem solved}
In this section, we assume that the nontrivial move problem is solved.

\iffull
\else

\fi
\begin{lemma}\labell{l:move:comm}
If the nontrivial move problem is solved, the direction agreement problem
can be solved in $O(1)$ rounds, also in the case that agents do not have
assigned IDs.
\end{lemma}
\iffull
\begin{proof}
Recall that, given a nontrivial move $t$, each agent can check whether
the rotation index of $t$ is larger than $n/2$ or smaller than $n/2$
according to its sense of direction using Lemma~\ref{l:opp:ind} (note that
the rotation index of a nontrivial move is not $0$ nor $n/2$).
Note that agents $a_1,a_2$ with opposite senses of direction
experience rotation indexes $r$ and $n-r$ respectively for some $r<n$.
Thus, the agents $a_1,a_2$ experience the rotation index of $t$ as larger than $n/2$
iff $a_1$ and $a_2$ have the same sense of direction. Therefore,
if agents experiencing the rotation index of $t$ as larger than $n/2$ change their
sense of direction, the direction agreement problem is solved.
\end{proof}
\fi
%


\iffull
For greater clarity, we provide a pseudocode of the above described solution as Algorithm~\ref{alg:agr:odd}. 
\else
\tj{The result stated in Lemma~\ref{l:move:comm} is obtained by the direction agreement protocol described in Alg.~\ref{alg:agr:odd}.}
\fi
\begin{algorithm}[]
	\caption{DirAgr($a$)}
	\label{alg:agr:odd}
	\begin{algorithmic}[1]
    \State Assign $\dir_a$ as in a nontrivial move \iffull\Comment{Assumption: nontrivial move solved}\fi
    \State \SingleRound
    \State $d_1\gets \pos()$
    \State \SingleRound
    \State $d_2\gets \pos()$
        \If{$d_1+d_2>1$}
            \State change sense of direction
        \EndIf
    \end{algorithmic}
\end{algorithm}
%
%
%

\begin{lemma}\labell{l:ntm:le}
Assume that the nontrivial move problem is solved.
Then, it is possible to solve the leader election problem
in $O(\log N)$ rounds.
\end{lemma}
\iffull
\begin{proof}
The idea of our solution is as follows.
Given a nontrivial move, we solve the direction agreement
problem in time $O(1)$ (Lemma~\ref{l:move:comm}).
Then, let $X$ be the set of agents which start moving right
in a known nontrivial move
(according to agreed common sense of direction).
Thus, $0<|X|<n$ and we choose $X$ as the initial set of candidates
for the leader.
Then we iterate over consecutive bits of IDs and gradually decrease
the set of candidates by fixing consecutive bits of candidates.
More precisely, in the $i$th step, we check whether the rotation
index of $X_0=\{b\,|\, b\in X, \ID_b[i]=0\}$ or
$X_1=\{b\,|\, b\in X, \ID_b[i]=1\}$ is non-empty.
By Lemma~\ref{l:ri}.2, at least one of them is nonempty
and we restrict $X$ to this subset ($X_0$ or $X_1$). In this way, we eventually
get a nonempty set of agents with all bits of IDs fixed, i.e., the set with one element
(the leader).
\end{proof}
\else
\fi
\iffull
For greater clarity, we provide a pseudocode as Algorithm~\ref{alg:lead:nontrivial}.
\else
\tj{
The result stated in Lemma~\ref{l:ntm:le} is obtained by the leader election protocol described in Alg.~\ref{alg:lead:nontrivial}.}
\fi
\begin{algorithm}[h]
	\caption{LeaderWithNMove($a$)}
	\label{alg:lead:nontrivial}
	\begin{algorithmic}[1]
		\State Solve the direction agreement problem\iffull\Comment{assumption: nontrivial move solved; see Lemma~\ref{l:move:comm}}\fi
    \State $X\gets$ all agents starting {\rright} in a nontrivial move
    \For{$i=1,2,\ldots\log N$}
            \State $X_0\gets\{b\,|\, b\in X, \ID_b[i]=0\}$\iffull\Comment{i.e., set $a\in X_0$ iff $a\in X$ and $\ID_a[i]=0$}\else\Comment{i.e., set $a\in X_0$ iff \phantom{0000000000000000000000000000}   $a\in X$ and $\ID_a[i]=0$}\fi
            \If{$\RI(X_0)\neq 0$} 
                \State $X\gets X_0$\Comment{i.e., set $a\in X$ iff $a\in X_0$}
						\Else
								\State $X\gets X\setminus X_0$\Comment{i.e., set $a\in X$ iff $a\not\in X_0$}
            \EndIf
    \EndFor
    \State Set the status of $a$ as {\em leader} iff $a\in X$.
    \end{algorithmic}
\end{algorithm}

\subsection{The setting with the chosen leader}
In this section, we assume that (exactly) one agent in a network has the
status ``leader''.
\begin{lemma}\labell{l:leader:move}
If the leader is chosen, one can solve the nontrivial move problem
in $O(1)$ rounds.
\end{lemma}
\begin{proof}
Assume that the leader $a$ is chosen. Consider two assignments of directions:
(1)~$\dir_b=$~right for each $b\in A$ and (2)~$\dir_b=$~right for each $b\neq a$ and
$\dir_a=$~left.
The rotation indexes $r_1,r_2$ of such two rounds differ by $2$ modulo $n$ (Lemma~\ref{lem:round}).
As $n>4$, at least one of two numbers which differ by $2$ modulo $n$ does not belong to
$\{0,n/2\}$. Thus, the nontrivial move problem is solved.
\end{proof}

\begin{corollary}\labell{c:leader:comm}
If the leader is chosen, one can solve the direction agreement problem
in $O(1)$ rounds.
\end{corollary}
\begin{proof}
Given the leader, we obtain a nontrivial move in $O(1)$ rounds (Lemma~\ref{l:leader:move}).
Next, we apply the solution from Lemma~\ref{l:move:comm} to obtain a common
sense of direction in $O(1)$ rounds.
\end{proof}

\subsection{The setting with the common sense of direction}
In this section we consider the setting that agents have the common sense
of direction. 
\iffull
We show simple efficient solutions for leader election and 
nontrivial move in the basic model which rely on the subroutine for emptiness
testing presented at the beginning of the section.
\else
We show simple efficient solutions for leader election and 
nontrivial move in the basic model which rely on the emptiness
testing result from the following lemma.
\fi
%
\begin{lemma}\labell{lem:emptiness}
Assuming all agents share a common sense of direction,
the emptiness testing problem
can be solved in $\log N$ rounds in the basic model, 
and in
one round in the {\lazy} and {\perceptive} model. Moreover,
if $n$ is odd, the emptiness testing is solvable in one round
in the {\basic} model as well.
\end{lemma}
\iffull
\begin{proof}
First, assume that 
the considered
model is lazy/perceptive or $n$ is odd.
In order to test whether $B'=B\cap A$ is empty:
\begin{itemize}
\item
every $a\in B$ moves right,
\item
other agents
move left in the basic model and perceptive model 
(they cannot choose ``idle'' in these models);
and they start as idle agents in the lazy model.
\end{itemize}
If the agents' positions at the end of the round with these directions are different
from their starting positions, then $B'\neq\emptyset$.
Otherwise,
\begin{itemize}
\item
in the basic/perceptive model: either $B'=\emptyset$, or $|B'|=n/2$ or $|B'|=n$ .
\item
in the lazy model: either $|B'|=0$ or $|B'|=n$.
\end{itemize}
As for the distinction between the cases $|B'|=0$ and $|B'|=n$, notice that each
agent $a\in A$ knows whether $\ID_a\in B'$. Thus, it can also distinguish
whether $B'=\emptyset$: if $\ID_a\in B'$ then $|B'|=n$ and $|B'|=0$ otherwise.
This settles the problem in the {\lazy} model, where we knew that either
$|B'|=0$ and $|B'|=n$
(when the position of an agent at the beginning of the round and at the
end of the round are equal).

In the {\basic} and {\perceptive} model, it might be the case that $|B'|=n/2$
and thus the elements from $A\setminus B'$ are still not able to distinguish
the cases $|B'|=0$ and $|B'|=n/2$:
However,
\begin{itemize}


\item
Perceptive model:

If $|B'|=n/2$, then each agent has at least one collision during a round,
while there are no collisions if $|B'|=0$ or $|B'|=n$ (as all agents start the round
with the same direction). Thus, in the perceptive model, an agent can distinguish
the case $|B'|=n/2$ from $|B|\in\{0,n\}$ by observing, whether it has noticed
a collision at all during a round.

\item
Basic model:

Here, we assume that $n$ is odd and therefore the case $|B'|=n/2$ does not happen.
Thus, each agent $a$ knows that $B'\neq\emptyset$ if and only if $\ID_a\in B$.
\end{itemize}

Regarding even $n$ and the basic model, we apply Algorithm~\ref{alg:empty:odd}.
It relies on the fact that, if $|B'|=n/2$, there is 
$i\in[\log N]$ and $j\in\{0,1\}$ such that
$0<|\{a\,|\, a\in B', \ID_a[i]=j\}|<n/2$.
%
\begin{algorithm}[]
	\caption{Emptiness($a$, $B$)\Comment{assumption: even $n$, basic model, comm.\ direction}}
	\label{alg:empty:odd}
	\begin{algorithmic}[1]
    \State If $a\in B$ then $\dir_a\gets \rright$ else $\dir_a\gets \lleft$
    \State \SingleRound
    \State If $\pos()\neq 0$: return $B\cap A$ is not empty
    \For{$i=1,2,\ldots\log N$}
            \State If $a[i]=1$ and $a\in B$ then $\dir_a\gets \rright$ else $\dir_a\gets\lleft$
            \State \SingleRound
            \State If $\pos()\neq 0$: return $B\cap A$ is not empty
    \EndFor
    \State If $a\in B$ return $B\cap A$ is not empty else return $B\cap A$ is empty.
    \end{algorithmic}
\end{algorithm}
As before, let $B'=B\cap A$.
Note that if $\pos()\neq 0$ after the first round, then $B'$ is certainly not empty.
Otherwise, $|B'|\in\{0,n,n/2\}$. Each agent $a\in B$ knows that $B'\neq\emptyset$,
so we only need to allow the agents outside $B$ to
distinguish between $|B'|=0$ and $|B'|=n/2$ (note that no agents outside $B$ attend
the protocol if $|B'|=n$). If $|B'|=n/2$, then there exists a bit $i\in[\log N]$ such that
$0<\{a\,|\, a\in X', a[i]=1\}<n/2$, and therefore the $i$th round in the for-loop
gives a nonzero rotation index, i.e., $\pos()\neq 0$ for each agent after this round.
On the other hand, if $|B'|=0$, each round of the for-loop will give the rotation
index $0$.
(Note that, as before, from the point of view of an external observer, the cases
$|B'|=0$ and $|B'|=n$ might be indistinguishable, since it is possible that
each round gives rotation index $0$ when $|B'|=n$. However, as agents know that
$|B'|\neq n/2$ and $|B'|\in\{0,n\}$, it is enough that they check if their own ID
is in $B$.)
\end{proof}
\else
\fi

With help of the emptiness testing protocol, we devise a solution to the leader
election problem. The idea of our solution is based on a binary search approach
similar to that from Lemma~\ref{l:ntm:le}. The main obstacle here
is that, without a nontrivial move, the initial set of candidates for the leader
is just $X=A$ and it has size $n$, thus its rotation index is $0$. And, the case that it is split in two
subsets $X_1,X_2$ of size $n/2$ is indistinguishable from the case that it is split in 
$X_1=X$ and $X_2=\emptyset$ (or vice versa), at least on the basis of rotation indexes of
appropriate sets. Therefore, we use the more sophisticated emptiness testing
from Lemma~\ref{lem:emptiness}.
\begin{lemma}\labell{lem:emptiness:leader}
Assuming all agents share common sense of direction,
the  leader election problem can be solved
in $O(\log^2N)$ rounds in the {\basic} model (with even $n$)
and in $\log N$ rounds in other settings.
\end{lemma}
\iffull
\begin{proof}
We implement a binary search
approach with help of the emptiness testing, where each round fixes one bit in
the binary representation of a (nonempty) set of candidates for the leader.
For completeness we describe an algorithm which employs this idea as 
Algorithm~\ref{alg:lead:odd}.
\begin{algorithm}[]
	\caption{Leader($a$)\Comment{assumption: common sense of direction}}
	\label{alg:lead:odd}
	\begin{algorithmic}[1]
    \State $X\gets$ all agents\Comment{i.e., set $a\in X$}
    \For{$i=1,2,\ldots\log N$}
            \State $Y\gets\{b\,|\, b\in X, \ID_b[i]=0\}$\Comment{i.e., set $a\in Y$ iff $a\in X$ and $\ID_a[i]=0$}
            \If{$Y$ is not empty}\Comment{use Emptiness(a,Y) -- Lemma~\ref{lem:emptiness}}
                \State $X\gets Y$\Comment{i.e., set $a\in X$ iff $a\in Y$}
            \EndIf
    \EndFor
    \State Set the status of $a$ as {\em leader} iff $a\in X$.
    \end{algorithmic}
\end{algorithm}
The proof relies on the fact that $X\neq\emptyset$ after each iteration and
it contains element with fixed leftmost $i$ bits of their IDs after the $i$th
iteration of the for-loop. 
Thus, all bits of IDs of elements of $X$ are fixed after the
for-loop.
Thus, $|X|=1$ by uniqueness of IDs.
%
\end{proof}
\else
\fi
An efficient solution for the nontrivial move problem can be easily obtained
from Lemma~\ref{lem:emptiness:leader} and Lemma~\ref{l:leader:move}.
\begin{corollary}\labell{cor:move:common}
If all agents have the same sense of direction, the nontrivial move problem can be solved
in $O(\log^2N)$ rounds in the {\basic} model (with even $n$)
and in $\log N$ rounds in other settings.
\end{corollary}

\iffull
Below we show that the 
\else 
We note that the
\fi
nontrivial move problem can also be solved in $O(\log N)$
rounds in the {\basic} model with even $n$, thus strengthening the $O(\log^2N)$
from Corollary~\ref{cor:move:common} for the basic model, and matching the bound
from this corollary for other models. However, the result in the following lemma
is weaker, as this is based only on a nonconstructive proof using the probabilistic
\iffull
method.
\else
method (omitted in this conference version).
\fi
\begin{lemma}\labell{lem:move:common}
If all agents have the same sense of direction, the nontrivial move problem can be solved in
$O(\log N)$ rounds.
\end{lemma}
\iffull
\begin{proof}
We randomly choose sets $S_1,\ldots,S_k$
such that each $a\in[N]$ is an element of $S_i$ with probability $1/16$, where all choices
are independent.
Assume that $n\geq 16$.
Then, in round $i$, the agents from $S_i$ move {\rright} and the remaining ones
move {\lleft}. The expected number of agents moving {\rright} is $n/16$ and,
by Chernoff bounds, the actual number of agents moving {\rright} in a round
is in $[n/32,3n/32]$ with probability $\ge 1-2^{-\Theta(n)}$. Thus, the $i$th round
gives a nontrivial move with probability $1-2^{-cn}$ for some constant $c$.
Let $T_{N,k}$ be an event that
a sequence $S_1,\ldots,S_k$ does not give a nontrivial move for at least one
set $X\subset[N]$ such that $|X|\geq 16$.
Then,
$$\prob(T_{N,k})< \sum_{j=16}^N{N\choose j}2^{-ckj}\leq \sum_{j=16}^N 2^{j(\log(Ne/j)-ck)}$$
If $k>2(\log N)/c$, then $\log(Ne/j)-ck\leq -1$ and
$$\prob(T_{N,k})<\sum_{j=16}^N 2^{-j}<1/2^{15}.$$

Finally, we inspect the case that $4<n<16$.
Note that there are only polynomially many subsets of $[N]$ of size at most
$15$.
Moreover, for each set $A$ of size $n$ in the range $(4,16)$ and each $i$, the probability
that $S_i$ gives a nontrivial move is at least some positive constant $c$,
independent of $N$ (e.g.\ this probability is not smaller than the probability
that exactly one element of $A$ belongs to $S_i$). 
%
Thus, one can derive a constant $c'$ such that each set $A$ of size in the interval
$(4,16)$ has a nontrivial move in $S_1,\ldots,S_{c\log N}$ with probability at least $1/2$.

The above facts give, by the probabilistic method, the result stated in the lemma.
\end{proof}
\else
\fi

\subsection{Application of coordination problems for location discovery}
Given the reductions summarized in Figure~\ref{fig:red:strong}
and Figure~\ref{fig:red:weak} (see also Theorem~\ref{ABC}), one can simply solve the location
discovery problem in the {\lazy} model, irrespective of the parity of $n$, or in
the basic model for odd $n$. 
\iffull
\else
\tj{This is the case, since given the common sense of direction and
the leader, we can obtain rotation index $1$ in the {\lazy} model and $2$ in the {\basic} model (all agents but the leader
move right at the beginning of a round).}
\fi 

\begin{lemma}\labell{l:from:disc}
Assume that (at least) one among the following problems is solved:
nontrivial move, leader election, direction agreement.
Then, location discovery can be solved in $n+O(\log N)$ rounds
in the {\lazy} model with arbitrary $n$ and in the {\basic} model
with odd $n$.
\end{lemma}
\iffull
\begin{proof}
Given a solution to any of the coordination problems, we can solve the leader
election problem and the direction agreement problem in $O(\log N)$ rounds 
(see Figure~\ref{fig:red:strong} and \ref{fig:red:weak}).

Then, in the lazy model, the leader starts every round moving right and other
agents start as idle. The rotation index is equal to $1$ and each agent knows
the relative positions of all other agents after $n$ rounds.

In the {\basic} model with odd $n$,  
each agent sets its direction to {\lleft}, except the leader which
sets the direction to {\rright}. The rotation index of each round is equal to $2$
in such a case which implies that if we repeat {\SingleRound} with this setting
several times, each agent gets back to its original position after exactly
$n$ rounds and it can determine the initial positions of all other agents at this moment,
on the basis of distances between its positions at the end of consecutive rounds.
Indeed, it we denote the distances by $x_1, x_2,\ldots, x_n$, where $x_1$ is the
distance to the closest agent in the direction {\rright}, then the consecutive
distances are $x_1+x_2,x_3+x_4,\ldots,x_n+x_1$ in the first $\lceil n/2\rceil$ rounds
and $x_2+x_3, x_4+x_5,\ldots, x_{n-1}+x_n$ in next $\lfloor n/2\rfloor$ rounds.
\end{proof}
\else
\fi
Note that the above result for the {\basic} model applies in the stronger
{\perceptive} model as well. However, we provide more efficient solutions for this
model later.

%

\subsection{Solutions for the case that $n$ is odd}\labell{sub:sec:odd}
\iffull
In this section we study the actual complexity of our considered problems
in the case that $n$ is odd.
\else
\fi

The crucial difference between the cases of odd and even $n$ follows from
the following observation:  If $n_C\neq 0$ and $n_A\neq 0$ in a round then
the round is nontrivial in the case of odd $n$.
On the other hand, this is not necessarily the case for even $n$,
as, e.g., $0\neq n_C=n_A=n/2$ or $n_C\in\{\frac34n,\frac14n\}$, $n_A=n-n_C$ do not give
a nontrivial move.
\begin{proposition}\label{prop:agree}
The direction agreement problem can be solved in $O(1)$ time in the basic
model, provided $n$ is odd.
\end{proposition}
\iffull
\begin{proof}
Observe that, if $n$ is odd, the rotation index of a round is zero only when all
agents start the round moving in the same direction (since $n/2\not\in\NAT$). 
Thus, one can ``test'' whether
all agents share the same sense of direction in a round in which each agent $a$ starts
with direction $\dir_a$=right.
Then if the rotation index of such round is zero, 
all agents have the same sense of direction.
Otherwise, this setting of directions gives a nontrivial move (as $n$ is odd).
Thus, we can apply Lemma~\ref{l:move:comm} to obtain a common sense of direction.
\end{proof}
\else
\fi

\begin{corollary}
If the number of agents $n$ is odd, the leader election problem and the 
nontrivial move problem can be solved in time $O(\log N)$.  The location
discovery problem can be solved in $n+O(\log N)$ rounds.
\end{corollary}
\iffull
\begin{proof}
The complexity of the leader election problem and the 
nontrivial move problem follows from 
Proposition~\ref{prop:agree} and Theorem~\ref{ABC}
while the complexity of location
discovery follows from Proposition~\ref{prop:agree} and Lemma~\ref{l:from:disc}.
\end{proof}
\else
\fi

\iffull
Below we provide a slightly modified variant of a solution for the nontrivial
move problem, reducing the complexity from $O(\log N)$ to $O(\log (N/n))$.
\else
There is also a slightly modified variant of a solution for the nontrivial
move problem, reducing the complexity from $O(\log N)$ to $O(\log (N/n))$.
\fi
\begin{proposition}\label{prop:nontrivial}
The nontrivial move problem can be solved in $\Theta(\log(N/n))$ time in
the {\basic} model with odd $n$.
\end{proposition}
\iffull
\begin{proof}
Since $n$ is odd, the rotation index cannot be equal to $n/2$ and thus it 
is sufficient to obtain any rotation index not equal to $0$.
As observed in Proposition~\ref{prop:agree}, an assignment $\dir_a=$right for
each $a$ gives a nontrivial move if and only if the agents do not share a common sense
of direction. Thus, in a round with such assignment of directions, either
we obtain a nontrivial move (problem solved) or we know that all agents share
the same sense of direction.
In the latter case, it is sufficient to split $A$ in two nonempty subsets $A_1$, $A_2$ such that
the elements of $A_1$ start a round moving right and 
the elements of $A_2$ start a round moving left.
We split agents on the basis of consecutive bits
of the binary encoding of their IDs (of length $\lceil\log N\rceil$).
We succeed after finding $i$ such that the set of agents with $0$ on the $i$th bit 
and the set of agents with $1$ on the $i$th bit are nonempty.
We use the observation that there are at most $\log (N/n)$ bits such
that all agents share the same value of IDs on those bits. Thus, assume
that the sets $\{a\,|\,\ID_a[i]=0\}$ and $\{a\,|\,\ID_a[i]=1\}$ are not empty
for some $i\in[\log N]$.
If agents
share the same sense of directions, and an agent $a$ chooses 
direction right if and only if $\ID_a[i]=0$,
then we get a nontrivial move after inspecting at most $\log(N/n)$ bits.
The following algorithm formalizes this idea.
\begin{algorithm}[]
	\caption{NMove($a$) \Comment{assumption: odd $n$}}
	\label{alg:move:odd}
	\begin{algorithmic}[1]
    \State $\dir_a\gets {\rright}$
    \State \SingleRound
    \State $d_1\gets \pos()$
    \If{$d_1\neq 0$}
        \State return
    \Else
        \State $X\gets$ all agents\Comment{i.e., set $a\in X$}
        \For{$i=1,2,\ldots$}
            \State $Y\gets\{b\,|\, b\in X, \ID_b(i)=0\}$
            \State if $a\in Y$ then dir$\gets \rright$ else dir$\gets\lleft$
            \State \SingleRound
            \State If $\pos()\neq 0$ return current value of dir \Comment{Nontrivial move}
        \EndFor
    \EndIf
    \end{algorithmic}
\end{algorithm}

As for the lower bound, assume that all agents share the same sense of direction and we have
an algorithm $\m{A}$.
Let $X_i$ be the largest set of IDs such that all elements
of $X_i$ choose the same direction in the first $i$ rounds of $\m{A}$, provided there is no nontrivial
round among rounds $1,\ldots,i-1$. Then, $|X_i|\geq N/2^i$. If $i<\log(N/n)$, then $|X_i|>n$ which
implies that $\m{A}$ does not give a nontrivial step if the set of IDs is equal to $X_i$.
\end{proof}
\fi



\section{Basic model with even $n$}\labell{sec:basic}
It is known (Lemma~\ref{l:disc:imp}) that the location discovery problem cannot 
be solved in the basic model when $n$ is even.
However, we can still try to solve other coordination problems. 
\iffull
We show in this section
that their complexity is significantly larger than for the case of odd $n$.
\else
Our results in this section state the their complexity is significantly larger than
for the case of odd $n$.
\fi
To this aim, we define a related combinatorial problem which we believe 
can be of independent interest.
\iffull
\else
Proofs are omitted here, but can be found in the Appendix.
\fi

First, we define a combinatorial notion of a {\em distinguisher}. Then, a relationship
between the size of a distinguisher and the 
complexity of the corresponding nontrivial move problem is established.
Finally, tight bounds on the smallest size of distinghuishers and the complexity
of the nontrivial move problem are showed.

\begin{definition}
We say that a family $\m{S}=\{S_1,\ldots,S_k\}$ of subsets of $[N]$ is a
{\em $(N,n)$-distinguisher} of size $k$ if for each $X_1,X_2\subseteq [N]$ such
that $|X_1|=|X_2|=n$ and $X_1\cap X_2=\emptyset$, there exists $i\in[k]$
such that $|S_i\cap X_1|\neq|S_i\cap X_2|$.
\end{definition}


\begin{definition}
Let $N\in\NAT$ and let $f\,:\,\NAT\times\NAT\to\NAT$ be a nondecreasing function.
A family $\m{S}=S_1,\ldots,S_{f(N,N)}$ of subsets of $[N]$ is a {\em strong $(N,f)$-distinguisher}
if the prefix $S_1,\ldots,S_{f(N,n)}$ of $\m{S}$ is 
a $(N,n)$-distinguisher for each $n\le N$.
\end{definition}

The {\em weak nontrivial move} problem is to assign to each agent $a$ a
direction $\dir_a$ such that if $a$ starts a round in the direction $\dir_a\in\{\rright,\lleft\}$,
then the rotation index $r$ in the round 
is not equal to $0$.
(A round with the rotation index $n/2$ is treated as a
weak nontrivial move, which is not the case in the standard definition of
a nontrivial move.)

\iffull
In what follows, we show 
\else
We first state
\fi 
a reduction between the complexity of the 
weak nontrivial move problem
and the smallest size of a distinguisher.  
\iffull
\else
\fi

\begin{proposition}\labell{prop:reduction}
Let $n>4$ be an even number and $N\geq n$. 
\begin{enumerate}
\item
Assume that a protocol $\m{A}$ solves the weak nontrivial move problem in the basic
model in $O(f(N,n))$ rounds when the value of $n$ is {\em known} to the agents. Then, there
exists a $(N,n/2)$-distinguisher of size $O(f(N,n))$.
\item
Assume that a protocol $\m{A}$ solves the weak nontrivial move problem in the basic
model in $O(f(N,n))$ rounds when the actual value of $n$ is {\em unknown} to the agents.
Then, there
exists a strong $(N,f')$-distinguisher for $f'(N,n/2)=O(f(N,n))$.
\end{enumerate}
\end{proposition}
\iffull
\begin{proof}
First, assume that $n$ is known and $\m{A}$ solves the weak nontrivial move problem.
Observe that, until the first round of $\m{A}$ with a weak nontrivial move, the 
only information available
to each agent is that its starting position in a round is equal to its position at the end
of a round. Thus, its behavior can be defined by a sequence of sets $S_1, S_2, \ldots$,
such that the agent $a$ chooses direction {\rright} in round $i$ (provided no nontrivial
move appeared before) if and only if $a\in S_i$.
Let us fix which sense of direction is ``correct''.
Then, consider the situation in which the set of agents $X_1$ with the 
correct sense of direction and
the set of agents $X_2$ with the incorrect sense of direction satisfy $|X_1|=|X_2|=n/2$.
Let $m_1=|X_1\cap S_i|$, $m_2=|X_2\cap S_i|$.
Then,
the rotation index (mod $n$) in round $i$ is
$$\begin{array}{rcl}
(|X_1\cap S_i|+|X_2\setminus S_i|)-(|X_1\setminus S_i|+|X_2\cap S_i|) & = & 
(m_1+n/2-m_2)-(n/2-m_1+m_2)\\
&&2(m_1-m_2).
\end{array}$$
And therefore the $i$th round of $\m{A}$ gives a (weak) nontrivial move if and only 
if $2(m_1-m_2)\not\in\{0,n\}$, 
which implies $m_1\neq m_2$. On the other hand, $m_1\neq m_2$ is equivalent to the fact that
$S_i$ distinguishes $X_1$ and $X_2$.
In conclusion, the sequence $S_1,S_2,\ldots$ defining $\m{A}$ is a $(N,n/2)$-distinguisher.

For unknown $n$, the result follows from the above reasoning and the fact that $\m{A}$
has to tackle arbitrary even $n\leq N$ which reflects the difference between
a standard $(N,n)$-distinguisher and its strong counterpart.
\end{proof}
\else
\fi
%
\iffull
Now, we provide a lower bound on the size of a strong $(N,n)$-distinguisher with
a simple proof based on a counting argument
(a similar bound in another context was given e.g.\ in \cite{GrebinskiK00}).
Although this result is subsumed by Lemma~\ref{lem:weak}, we provide it 
to give some intuition before a more complicated, and less intuitive proof,
of Lemma~\ref{lem:weak}.
\begin{lemma}\labell{lem:strong}
If $S$ is a strong $(N,f)$-distinguisher for any $N>4$ 
and $f\,:\,\NAT\times\NAT\to\NAT$, then 
$f(N,n)=\Omega\left(\frac{n\log(N/n)}{\log n}\right)$.
\end{lemma}

\begin{proof}
First, we show  that a strong $(N,f)$-distinguisher  $\mathcal{S}$ 
satisfies the property that for each two different sets $X_1,X_2\subset [N]$ such that
$|X_1|=|X_2|=n$, there exists $i\leq f(N,n)$ such that
$|X_1\cap S_i|\neq |X_2\cap S_i|$
(note that $X_1$ and $X_2$ do not have to be disjoint!).
Indeed, assume to the contrary that this is not the case 
for $\mathcal{S}$, and thus $|X_1\cap S_i|=|X_2\cap S_i|$
for some different sets $X_1,X_2$ of size $n>1$ and each $i\in[f(N,n)]$.
Let $Y_1=X_1\setminus X_2$ and $Y_2=X_2\setminus X_1$. Then, $Y_1\cap Y_2=\emptyset$,
$|Y_1|=|Y_2|\leq n$ and $|Y_1\cap S_i|=|Y_2\cap S_i|$ for each $i\in[f(N,n)]$. This implies
that $\mathcal{S}$ is not a strong $(N,f)$-distinguisher, which is a contradiction.

Let $\mathcal{S}=(S_1,\ldots,S_k)$ be a strong $(N,f)$-distinguisher. The above observation implies
that, for any $X\neq X'$, $X,X'\subset[N]$ of size $n$, the sequences $|X\cap S_1|,\ldots,|X\cap S_k|$
and $|X'\cap S_1|,\ldots,|X'\cap S_k|$ are not equal, where
$k=f(N,n)$.
As each $S_i$ gives at most $n+1$ possible values of $|X\cap S_i|$ for $X\subset[N]$ of size $n$,
and there are $N\choose n$ subsets of $[N]$ of size $n$, we obtain the following bound
$$k\geq \log_{n+1}{N\choose n}=\Omega\left(\frac{\log{N\choose n}}{\log(n+1)}\right)=\Omega\left(\frac{n\log(N/n)}{\log n}\right)$$
for $n>1$.
\end{proof}
It turns out that the result of Lemma~\ref{lem:strong} can be strengthened, that is,
we show that the same asymptotic lower bound holds for a standard distinguisher as well.
However, our proof of this fact is much more complicated.  It applies
techniques from~\cite{ClementiMS03}, designed for proving lower bounds on size of
selective families.
We stress here that the lower bound for a strong variant of a distinguisher does not imply an analogous lower bound for
a ``standard'' variant of a distinguisher.
As observed in the proof of Lemma~\ref{lem:strong}, the
prefix of size $f(N,n)$ of a stong $(N,f)$-distinguisher
gives an opportunity to ``distinguish'' \textbf{each} pair of sets
of size $n$. On the other hand, a standard $(N,n)$-distinguisher
is supposed to give a difference only on \textbf{disjoint} sets
of size $n$.
\else
We now establish a lower bound on the size of a $(N,n)$-distinguisher in terms of
the parameters $N$ and $n$.  
\fi
\begin{lemma}\labell{lem:weak}
If $\mathcal{S}$ is a (standard) $(N,n)$-distinguisher for $N>2$ and $n \leq N/128$, 
then the size of $\mathcal{S}$ is $\Omega\left(\frac{n\log(N/n)}{\log n}\right)$.
\end{lemma}

\iffull
\begin{proof}
Let us first stress that the calculations from the previous lemma do not apply here, since
a (``standard'') distinguisher does not have to ``distinguish'' small sets, so it does not have to distinguish
non-disjoint sets of size $n$ either.
\else
\fi

\iffull
In the proof, we use a notion from \cite{ClementiMS03}:
\else
Our proof uses a notion from \cite{ClementiMS03}:
\begin{definition}\cite{ClementiMS03,FranklF85}
Let $l\leq k\leq n$. A family $\mathcal{F}$ of $k$-subsets (i.e.\ subsets of size $k$) of $[N]$ is
$(N,k,l)$-intersection free if $|F_1\cap F_2|\neq l$ for every $F_1,F_2\in\mathcal{F}$.
\end{definition}
\begin{fact}\labell{fac:free}\cite{ClementiMS03,FranklF85}
Let $\mathcal{F}$ be an $(N,k,k/2)$-intersection free family where $k$ is a power of $2$
and $k\leq N/64$. Then,
$$\log|\mathcal{F}|\leq \frac{11k}{12}\log(N/k).$$
\end{fact}
\iffull
Let $G(V,E)$ be a graph, whose vertices are all $2n$-subsets of $[N]$, where the 
edges connect vertices
corresponding to sets which have exactly $n$ common elements.
That is, $(X_1,X_2)\in E$ for $X_1,X_2\in V$ if and only if $|X_1\cap X_2|=n$.
Let $\alpha(G)$ and $\chi(G)$ denote the size of the largest independent set of $G$
and the chromatic number of $G$, respectively.
We claim that
\begin{eqnarray}
\log \chi(G) & \geq & \frac16 n\log(N/(2n)) \ \ \ \textrm {and}\label{e:color1}  \\
\log \chi(G) & \leq  & |\m{S}|\log(2n+1). \label{e:color2}
\end{eqnarray}
Proof of (\ref{e:color1}):\newline
We use the fact that $\chi(G)\geq \frac{|V|}{\alpha(G)}$.
Moreover, as each independent set of $G$ is a $(N,2n,n)$-intersection free family
of sets,
Fact~\ref{fac:free} implies that
$$\log\alpha(G)\leq \frac{22}{12}n\log(N/(2n)).$$
Therefore
$$\begin{array}{rclcl}
\log \chi(G) & \geq & \log |V|-\log\alpha(G) \\
                & \geq & \log {N\choose 2n}-\frac{22}{12}n\log(N/(2n)) \ \ \ \textrm {and} \\
               & \geq & 2n\log(N/(2n))-\frac{22}{12}n\log(N/(2n)) \\
               & = & \frac{1}{6}n\log(N/(2n)), 
\end{array}
$$
which gives (\ref{e:color1}). In the third inequality, we use the relation 
${a \choose b}\ge \left(\frac{a}{b}\right)^b$.

\noindent Proof of (\ref{e:color2}):\newline
\comment{
Let $\m{S}=(S_1,\ldots,S_m)$ be a $(N,n)$-distinguisher, and let $G_i(V_i,E_i)$ be a graph associated with $S_i$ such that
\begin{enumerate}
\item
$V_i=V$ and
\item
$(X_1,X_2)\in E_i$ if and only if $(X_1,X_2)\in E$ and
$|(X_1\setminus X_2)\cap S_i|\neq |(X_2\setminus X_1)\cap S_i|$
(i.e.\ $S_i$ distinguishes $X_1\setminus X_2$ and $X_2\setminus X_1$).
\end{enumerate}
Now, observe that  the set $\{X\,|\ |S_i\cap X|=p\}$ is independent in $G_i$
for each $p\in[0,n]$. 
Therefore, $\chi(G_i)\leq n+1$ since
the graphs $G_1,\ldots,G_m$ share the same set of vertices, so
$\chi(\bigcup_{i=1}^m G_i)\leq \Pi_{i=1}^m \chi(G_i)$.
The assumption that $\m{S}$ is a $(N,n)$-distinguisher implies that $\bigcup_i E_i=E$.
Therefore,
$$
\begin{array}{rcl}
\log \chi(G)  & \leq & \log \chi(\bigcup_i G_i) \\
                 & \leq & \log\left\{ \Pi_{i=1}^m \chi(G_i)\right\}  \\
                 & =   &\sum_{i=1}^m\log\chi(G_i)\leq m\log(n+1),
\end{array}
$$
which proves (\ref{e:color2}).
}
Let $\m{S}=(S_1,\ldots,S_m)$ be a $(N,n)$-distinguisher.
Observe that for any two sets $X_1,X_2$ such that $|X_1\cap X_2|=n$ there exists $S_i$
such that $|S_i\cap X_1|\neq |S_i\cap X_2|$.
In other words for any tuple $(p_1,\ldots,p_m),p_i\in[0,2n]$
the set $\{X\,:\forall i \ |S_i\cap X|=p_i\}$ is independent in $G$.
Therefore, $\chi(G)\leq (2n+1)^m$.
Thus
$$
\log \chi(G)  \leq m\log(2n+1),
$$
which proves (\ref{e:color2}).

Finally, observe that (\ref{e:color1}) and (\ref{e:color2}) imply the statement
of the lemma.
\end{proof}
\else
\fi
\begin{corollary}\labell{cor:lower:move}
Each algorithm solving the (weak) nontrivial move problem 
requires $\Omega(n\log(N/n)/\log n)$ rounds in the {\basic}
model with known value of $n$.
\end{corollary}
\iffull
\begin{proof}
The result 
follows directly from Proposition~\ref{prop:reduction}
and Lemma~\ref{lem:weak}.
\end{proof}
\else
\fi
\iffull
Now, we 
show, using the probabilitstic method, 
\else 
It can be shown using the probabilistic method
\fi
that there exists a solution for the 
nontrivial move problem that nearly matches the lower bound from 
Corollary~\ref{cor:lower:move}.
\begin{theorem}\labell{thm:selector}
In the {\basic} model, there exist solutions of the nontrivial move problem working in $O(n\log(N/n)/\log n)$ rounds
for each $n\in[N]$ and $n>4$, 
and also when $n$ is unknown. 
\end{theorem}
\iffull
\begin{proof}
Let us choose a sequence $\m{S}$ of sets $S_1,S_2,\ldots$ probabilistically, such that each
$x\in[N]$ belongs to $S_i$ with probability $1/2$, where all choices are independent.
Then, our algorithm is defined such that, in round $i$, the agents with IDs in $S_i$
choose direction {\rright} and the other ones choose the direction {\lleft}.
We show that the family $\m{S}=(S_1,\ldots,S_k)$ chosen in this way gives a protocol
solving the nontrivial move problem with positive probability, provided the size
$n$ of the network is smaller than $N/3$. 
That is, the following event holds with positive probability:
for each $X\subset [N]$ such that $|X|<N/3$, the nontrivial move appears during
an execution of the prefix of $\m{S}$ of size $O(n\log(N/n)/\log n)$, where $n=|X|$.
Then we build a sequence $\m{C}$ of size $O(N/\log N)$ which gives a nontrivial move
on each $X\subset [N]$ of size at least $N/3$.
Thus, by interleaving $\m{S}$ and $\m{C}$,  the theorem holds thanks to the probabilistic
method.

Let us fix a set of IDs $A\subset[N]$ of size $n$ and assign sense of directions to them
such that $A=A_c\cup A_i$, where $A_c$ is the set of agents with correct sense
of directions, $|A_c|=n_c$ and $|A_i|=n-n_c$.
Recall that a round does {\em not} give a nontrivial move if and only if it is a
$(0,n)$-round, $(n,0)$-round, $(n/2,n/2)$-round, $(3n/4,n/4)$-round,
or a $(n/4,3n/4)$-round. Then, 
for a round defined by $S_i$ as above, we have:
$$\begin{array}{rcll}
\text{Prob}((n/2,n/2)\text{-round})&=&\frac1{2^n}\sum_{j=0}^{\min(n_c,n/2)} {n_c\choose j}{n-n_c\choose n/2-j} \\
                                               &=&\frac1{2^n}{n\choose {n/2}}\leq \frac{c_0}{n^{1/2}}, \\ 
\text{Prob}((0,n)\text{-round})&=&\frac1{2^n}{n_c\choose 0}{n-n_c\choose n-n_c} =\frac1{2^n}, \\
\text{Prob}((n,0)\text{-round})&=&\frac1{2^n}{n_c\choose n_c}{n-n_c\choose 0} =\frac1{2^n}, \\
\text{Prob}((n/4,3n/4)\text{-round})&=&\frac1{2^n}\sum_{j=0}^{\min(n_c,3n/4)} {n_c\choose j}{n-n_c\choose 3n/4-j} \\
                                                 &=&\frac1{2^n}{n\choose n/4}=1/2^{\Theta(n)}, \ \textrm{and} \\
\text{Prob}((3n/4,n/4)\text{-round})&=&\text{Prob}((n/4,3n/4)\text{-round})  \\
                                                 &=&1/2^{\Theta(n)}.\\
\end{array}$$
In the above calculations, we use the relationship that
$\sum_{i=0}^{\min(a,c)} {a\choose i} {b-a\choose c-i}={b\choose c}$ and Stirling's formula 
which determines the constant $c_0$ in the first row.
The above estimations imply that the probability that a round defined by $S_i$
is a trivial move for $|A|=n$ is at most $c_1/\sqrt{n}$ for some 
constant $c_1$, provided $n$ is large enough.
Let us consider all sets of IDs $A$ such that $|A|\in[2^{i-1},2^i)$, 
for $i$ such that $2^i<N/3$.
Let $k=c\frac{2\log{N\choose 2^i}}{i-1}$ for a large enough constant $c$ whose value
will be determined later.
By $E_i$ we denote the event that a sequence of sets $S_1,\ldots,S_k$ does {\em not}
give a nontrivial move for
all sets $A$ whose size is in $[2^{i-1},2^i)$. Then, 
$$\begin{array}{rclcl}
\text{Prob}(E_i) &\leq& \sum_{d=2^{i-1}}^{2^i}(\text{Prob}(\text{triv.\ move on a set of size }d))^k\cdot{N\choose d} 2^d \\
&\leq&
\sum_{d=2^{i-1}}^{2^i}\frac{{N\choose d} 2^d\cdot c_1}{2^{(i-1)k/2}}
\leq c_1\sum_{d=2^{i-1}}^{2^i}\frac{{N\choose d}^2}{{N\choose 2^i}^3}\\
&\leq& c_1\sum_{d=2^{i-1}}^{2^i}\frac{1}{{N\choose 2^i}}\leq
c_1\sum_{d=2^{i-1}}^{2^i}\frac{1}{2^i}<c_1\frac1{2^{i-2}}.
\end{array}$$
%
In the above calculations, we use the following facts:
\begin{itemize}
\item
${N \choose d}2^d$ is the number of possible choices of sets of size $d$, and senses of direction of elements of these sets (used in the first inequality);
\item
$\text{Prob}(\text{triv.\ move on a set of size }d)\leq \frac{c_1}{\sqrt{d}}\leq \frac{c_1}{2^{(i-1)/2}}$ 
(used in the second inequality);
\item
$2^{(i-1)k/2}\geq {N\choose 2^i}^{c}\ge {N\choose 2^i}^{3}$ for $c\ge 3$ (which 
follows from the fact that $k=c\frac{2\log{N\choose 2^i}}{i-1}$; used in the third
inequality);
\item
$2^d\leq {N\choose d}$
for $d\leq N/3$ (used in the third inequality);
\item
${N\choose d}\leq {N\choose 2^i}$  
for $d\leq N/3$ (used in the fourth inequality);
\item
${N\choose 2^i}\ge (N/2^i)^{2^i}\ge 2^i$ if $2^i<N/2$ (used in the fifth inequality).
\end{itemize}
Let $i_0=\lceil \log 4c_1\rceil+1$ and $i_1=\lfloor \log (N/3)\rfloor$.
The above calculations show that, the union of events $E_{i_0}, E_{i_0+1},\ldots,E_{i_1}$
holds with probability $\sum_{i_0}^{i_1}c_1/2^i<1/2$ for $c>3$.
Therefore, by the probabilistic method, the sequence $\m{S}$ gives a nontrivial move for each set of IDs of size
in $[2^{i_0},2^{i_1}]=[4c_1,N/c]$. It remains to tackle the cases that $n<2^{i_0}$ and $n>2^{i_1}$.

As for $n<2^{i_0}$, note that $2^{i_0}$ is a constant independent of $n$.
Thus the number of sets of size $<2^{i_0}$ is polynomial wrt $N$, while the
probability that a round gives a nontrivial move for a given set is larger
than some positive constant independent of $N$. Therefore on a sufficiently
long prefix of $\m{S}$ of length $O(\log N)=O(n\log(N/n)/\log n)$, the nontrivial
move appears with for each set of size
$<2^{i_0}$ with probability $1-1/N$.

Now, we consider the case that the size $n>2^{i_1}>N/3$. 
The number of such sets is upper bounded by $2^N$.
And, for each such set, each round gives a nontrivial move with probability at least
$c'/\sqrt{N}$ for a constant $c'$. By a simple calculation, one can show that
the nontrivial move appears for each such set on a long enough prefix
of $\m{S}$ of size $O(N/\log N)$ with probability $1-1/N$.
More precisely, on a prefix of $c''\log N$, the probability that there is a
set without a nontrivial move is smaller than
$$2^N (c'/\sqrt{N})^{c''\log N}<1/N$$
for a large enough constant $c''$.

\end{proof}
\else
\fi
\iffull
Finally, the above bounds (Cor.~\ref{cor:lower:move}, Th.~\ref{thm:selector})
and the equivalence of complexities of coordination problem (Th.~\ref{ABC})
lead to the following corollary.
\else
\comment{
\tj{
Cor.~\ref{cor:lower:move}, Th.~\ref{thm:selector}
and 
Th.~\ref{ABC}
lead to the following corollary.
}
}
\fi
%
\begin{corollary}\labell{cor:coord:basic}
The time complexity of 
the nontrivial move problem, the leader election problem,
 and the direction agreement problem
in the basic model (with even $n$)
is $\Theta(n\log(N/n)/\log n)$.
\end{corollary}
\iffull
\else
\tj{The above result follows from Cor.~\ref{cor:lower:move}, Th.~\ref{thm:selector}
and 
Th.~\ref{ABC}.
}
\fi
Given the relationship between distinghuishers  and the nontrivial move problem
(Prop.~\ref{prop:reduction}), the lower bound 
from Lemma~\ref{lem:weak}
and Cor.~\ref{cor:coord:basic}, we get the following bound. 
\begin{corollary}\labell{cor:dist:size}
The size of the smallest $(N,n)$-distinguisher for $N\ge n$ is $\Theta(n\log(N/n)/\log n)$.

For each $N\in\NAT$, there exists a strong $(N,f)$-distinguisher for some $f(N,n)\in O(n\log(N/n)/\log n)$. Moreover, if $\m{S}$
is a strong $(N,f)$-distinguisher, then $f(N,n)=\Omega(n\log(N/n)/\log n)$.
\end{corollary}

\iffull
\section{Lazy model with even $n$}\labell{sec:lazy}
\input{lazy.tex}
\else
\tj{We note that the bound from Cor.~\ref{cor:coord:basic} also holds for the lazy model 
(the proofs are omitted). 
It follows from the fact that complexities of
the weakly non-trivial move in the {\basic} model and 
the non-trivial in the {\lazy} are asymptotically equal.
}
\fi

\section{Perceptive model without common sense of direction}\labell{sec:perc}
Since the \basic\ model is too weak for the task of position discovery (when $n$ is even),
we considered the \lazy\ model. Although one can solve position discovery in this model, the
overhead cost for this problem is $\Omega(n\log(N/n)/\log n)$.
In \cite{FriedetzkyGGM12}, it is shown that position discovery
can be solved in the {\perceptive} model (i.e., when the position
of the first collision in a round can be detected while each
agent has to start the round moving to the {\rright} or {\lleft}).
%
In this section, we inspect efficiency of coordination problems as well as position discovery in this model.
First, we show that the {\perceptive} model gives an opportunity to exchange information between neighbors
on a ring (Section~\ref{sub:sec:comm}). Then, we use this feature to build algorithms for the nontrivial move problem 
which brake the lower bounds working in the {\basic} model and the {\lazy} model (Section~\ref{sub:sec:move:perc}).
Finally, using these solutions as tools, we provide a solution for the positions discovery problem
in time $n/2+o(n)$ provided $\log N=o(\sqrt{n})$ which is optimal up to the $o(n)$ term (Section~\ref{sub:sec:disc:perc}).

\subsection{Communication on a ring}\labell{sub:sec:comm}
First, we discuss the following {\em neighbors discovery} task in which
each agent $a$ should:
\begin{itemize}
\item
\vspace{-5pt}
learn (relative) location of its left neighbor $\Lleft(a)$ and its right neighbor $\Rright(a)$;
\item
\vspace{-5pt}
determine whether $\Lleft(a)$ and $\Rright(a)$ have the same sense of direction as $a$ has.
\end{itemize}
\vspace{-5pt}
Algorithm~\ref{alg:neighbor} solves this problem based on the fact that each two
IDs differ on at least one bit. (Some calculations performed by agents are not
explicitly described in the algorithm, they are discussed later.)
In Algorithm~\ref{alg:neighbor}, each execution of {\SingleRound} is followed by {\ReversedRound} in which
each agent starts a round with the direction opposite to its local direction {\dir}. We omit this detail in the pseudocode. However, let us stress here that this gives a guarantee that each agent starts each application of {\SingleRound} at exactly the same position as its position before the execution of the
algorithm (so, its distances to neighbours are the same as well). 

\begin{algorithm}[]
	\caption{NeighborDiscovery($a$)}
	\label{alg:neighbor}
	\begin{algorithmic}[1]
    \State $D_{\lleft}\gets\emptyset$; $D_{\rright}\gets\emptyset$ \Comment{distances to collisions}
    \For{$i=1,2,\ldots,\log N$}
        \For{$j\in[0,1]$}
            \For{$k\in\{\lleft,\rright\}$}
                \If{$\ID_a[i]=j$} $\dir\gets k$
                \Else   \ $\dir\gets$ direction opposite to $k$ 
                \EndIf
                \State \SingleRound
                \State After a round: $D_k\gets D_k\cup\{\coll()\}$ 
            \EndFor
        \EndFor
    \EndFor
    \State $\dir_a\gets \rright$ \Comment{All agents choose direction {\rright}}
    \State \SingleRound; $D_\rright\gets D_\rright\cup\{\coll()\}$ 
    \State $\dir_a\gets \lleft$
    \State \SingleRound; $D_\lleft\gets D_\lleft\cup\{\coll()\}$ 
    \State Location of $\Rright(a)\gets 2 \min (D_{\rright})$
    \State Location of $\Lleft(a)\gets 2 \min (D_{\lleft})$
    \end{algorithmic}
\end{algorithm}

\begin{proposition}\label{p:alg1}
Algorithm~\ref{alg:neighbor} gives solution to neighbors discovery in $O(\log N)$ rounds.
\end{proposition}
\iffull
\begin{proof}
Consider $a$ and $a'=\Rright(a)$. First, we show that, in some round,
they start moving towards each other (which gives the distance to collision
equal to halve of their distance, the smallest possible). If they have {\em opposite} sense of direction,
then this happens in line 10 of the algorithm. Otherwise (i.e., they have the
same sense of direction), they start moving towards every other for such $i,j,k$
that $\ID_a[i]\neq \ID_{a'}[i]$, $j=\ID_a[i]$ and $k=\rright$.

When $a$ already knows its distance to $\Rright(a)$, the distance to collision
in line 10 gives information whether $a$ and $\Rright(a)$ have the same sense of direction.
(Namely, their senses of directions are opposite iff the result of $\coll()$ is equal to halve of
the distance between them.)
\end{proof}
\else
\fi

\comment{
\begin{proposition}\label{p:alg1}
One can modify Algorithm~\ref{alg1} such that each agent learns ID of its left neighbor
and ID of its right neighbor in $O(\log N)$ rounds.
\end{proposition}
\begin{proof}
Recall that, after Algorithm~\ref{alg1}, each agent knows sense of direction of its neighbors
and distances $d_l,d_r$ to them.
Let $a_1$ be an agent, and $a_2=\Rright(a_1)$.
Observe that, in Algorithm~\ref{alg1}, $a_1$ and $a_2$ start moving towards
each other in line 7 for particular $i$ iff:
\begin{itemize}
\item
$a_1$ and $a_2$ have the same sense of direction, $\ID_{a_1}[i]\neq \ID_{a_2}[i]$, $j=\ID_{a_1}[i]$, $k=\rright$; OR
\item
$a_1$ and $a_2$ have opposite sense of direction, $\ID_{a_1}[i]= \ID_{a_2}[i]$, $j=\ID_{a_1}[i]$, $k=\rright$.
\end{itemize}
\end{proof}
}

\begin{proposition}\label{prop:communication}
If each agent knows:
\begin{itemize}
\item
\vspace{-5pt}
locations of its neighbors (relative to its initial location); AND
\item
\vspace{-5pt}
sense of direction of its neighbors (with respect to its own sense of direction);
\end{itemize}
then each agent can transmit one bit of information
to its neighbors in time $O(1)$.
\end{proposition}
\iffull
\begin{proof}
Assume that each agent $a$ is going to transmit bit $z_a$. Consider a {\em phase}
which consists of four rounds:
\begin{itemize}
\item Round $1,2$: if $z_a=1$ then $\dir_a\gets \rright$ else $\dir_a\gets \lleft$;

\SingleRound, \ReversedRound

\item Round $3,4$: if $z_a=0$ then $\dir_a\gets \rright$ else $\dir_a\gets \lleft$;

\SingleRound, \ReversedRound
\end{itemize}
Consider an agent $a$ and its right neighbor $b=\rright(a)$. We show that $a$ is able
to determine $z_b$ (the bit transmitted by $b$) based on results of Rounds $1$ and $3$ (i.e., on moments of first
collisions). We say that a round is {\em clear} if $a$ and $b$ collide in the middle point
between their original locations, which means that they start a round moving towards
each other.
We have two cases:

\noindent Case 1: the sense of direction of $a$ and $b$ agree:
 \begin{enumerate}
 \item[(a)]
 $z_a=1$: if Round 1 is clear, then $z_b=0$, otherwise $z_b=1$.
 \item[(b)]
 $z_a=0$: if Round 3 is clear, then $z_b=1$, otherwise $z_b=0$.
 \end{enumerate}

\noindent Case 2: the sense of direction of $a$ and $b$ do {\em not} agree:
 \begin{enumerate}
 \item[(a)]
 $z_a=1$: if Round 1 is clear, then $z_b=1$, otherwise $z_b=0$.
 \item[(b)]
 $z_a=0$: if Round 3 is clear, then $z_b=0$, otherwise $z_b=1$.
 \end{enumerate}

Communication with $\lleft(a)$ can be performed in an analogous way.

\end{proof}
\else
\tj{\vspace{-5pt}
The statement of Prop.~\ref{prop:communication} can be obtained such that each
agent starts  round 1 (2, resp.) moving left/right depending on the transferred  bit. Then, the distances to the first collision in both rounds give information about the bits of  neighbors.}
\fi
Since agents can learn location of their neighbors and their sense(s) of direction
in $O(log N)$ rounds (see Proposition~\ref{p:alg1}), Proposition~\ref{prop:communication}
leads to the following corollary.

\begin{corollary}\labell{c:sim}
There exists a possibility to exchange one bit of information
between each two neighbors in the {\perceptive} model  in time $O(1)$, 
after a $O(\log N)$ preprocessing.
\end{corollary}

The above corollary gives opportunity to simulate any distributed algorithm on a ring
in message passing model (i.e., when each pair of neighbors can exchange a message in
one round of computation).
%
%
However, the time efficiency of such simulations is limited by the fact that only one bit
of information is exchanged between neighbors in a round.
%

Let {\em information dissemination task} with parameters $d$ and $p$ be to disseminate
a message $m_a$ with $p$ bits by each agent $a$ to all agents in ring distance $\leq d$
from $a$.
\begin{corollary}\labell{cor:disseminate}
Information dissemination task in which agents are supposed
to transmit messages of length $p$ on the ring distance $d$ can be accomplished in time $O(p\cdot d)$.
\end{corollary}
A solution claimed in the above corollary might we designed such that
first all agents transmit own messages, then messages arriving from their left neighbors and finally messages arriving from their right
neighbors.

Assume that $A'\subset A$ is a set of {\em marked} agents such that each agent
knows whether it is marked or not and the ring distance between any different $a,a'\in A'$
is at least $d$. Moreover, each $a\in A'$ has a message $M_a$ of size $\leq m$.
The {\em sparsed information dissemination task} with parameters $A',d$ and $m$
is to deliver the message of each $a\in A'$ to all agents in the ring distance
$\leq d$ from $a$. 
For an agent in $A'$, we denote this task by Diss($M_a,d$).
Using the procedure exchanging a bit of information between
each pair of neighbors in time $O(1)$, we obtain the following result.
\begin{corollary}\labell{cor:sdisseminate}
Sparsed information dissemination task in which agents in distances $\geq d$ are supposed
to transmit messages of length $p$ on the ring distance $d$ can be accomplished in time $O(p+d)$.
\end{corollary}
In a solution to the sparsed information dissemination we have to tackle
the fact that an agent has no direct way to convey a message of the type ``I have nothing to transmit (yet)''. One can solve this issue by a simple encoding, e.g., $00/11$ encodes $0/1$, while $01$ encodes ``no bit to transmit''.

\subsection{Nontrivial Move}\labell{sub:sec:move:perc}
As we know, the nontrivial move problem is intuitively to break balance between the number
of agents moving clockwise and anticlockwise.
In our solution we use $(N,k)$-selective families from \cite{ClementiMS03}.
\begin{definition}
Let $n<N$. A family $\m{F}$ of subsets of $[N]$ is 
$(N, n)$-selective if, for every non empty subset $Z$ of $[N]$ such that $|Z|\leq n$, there is a
set $F$ in $\m{F}$ such that $|Z \cap F|=1$.
\end{definition}
Clementi et al.\ \cite{ClementiMS03} showed that 
for any $N>2$ and $n\leq N$, there exists an $(N,n)$-selective family of size
$O(n \log(N/n))$.

Let a {\em local leader} for some fixed number $d$ be an agent $a$ with the largest ID among agents
in the ring distance $d$ from $a$.
In Algorithm~\ref{alg:move:perceptive} we present a solution to the nontrivial move
problem by establishing local leaders for exponentially growing distances $d=2^k$ and trying
to execute $(N,2^k)$-selective family on those leaders. As the number of local leaders is
$\leq n/2^k$, it becomes smaller than $2^k$ for $k>\frac12\log n$ and gives a nontrivial
move after $O(2^{\frac12\log n}\log N)=O(\sqrt{n}\log N)$ rounds.

\begin{algorithm}[]
	\caption{NMoveS($a$)}
	\label{alg:move:perceptive}
	\begin{algorithmic}[1]
    \State $\dir_a\gets \rright$; set the status of $a$ as a local leader;
		\State \SingleRound
    \State If the current directions give a nontrivial move: return
    \State Establish $1$-bit communication\Comment{Cor.~\ref{c:sim}}
    \For{$k=0,1,2,\ldots,\log N$}
        \State Sparsed dissemination of ID$_a$ of local leaders on distance $2^k$\Comment{Corollary~\ref{cor:sdisseminate}}
        \If{$\ID_a=\max(N_a(2^k))$}
            \State set the status of $a$ as the {\em local leader}
        \Else
            \State set the status of $a$ as {\em not leader}
        \EndIf
        \State Execute a $(N,2^k)$-selective family $\m{F}$ on the local leaders
        \State \textbf{if} a nontrivial move appears during execution of $\m{F}$ \textbf{then} return
    \EndFor
    \end{algorithmic}
\end{algorithm}

\begin{lemma}\labell{lem:ntmp}
The algorithm NMoveS solves the nontrivial move problem in $O(\sqrt{n}\log N/\log n)$ rounds
in the {\perceptive} model.
\end{lemma}
\iffull
\begin{proof}
As for correctness, the number of local leaders in the $i$th iteration of the
for-loop is $O(n/2^i)$. Thus, for $k\geq \log\sqrt{n}$, the number of local leaders
is not larger than $2^k$ and therefore the $(N,2^k)$-selective family breaks symmetry
between them which gives a nontrivial move. 


Below, we present the above intuition in more detail.
Let us fix some ``objective correct'' sense of direction.
Let $A_C,A_I$ be the subsets of agents with correct and 
incorrect sense of direction, respectively.
If the algorithm does not finish its execution in line 3,
the rotation index is in $\{0,n/2\}$ for a round with all $\dir_a$ equal to
{\lleft}. If exactly one agent changes its initial direction in
a round, the rotation index will increase by $2$ or $-2$.
That is, it will not be in $\{0,n/2\}$ since $n>4$,
and a nontrivial move will be obtained.
Let $L$ be the set of local leaders in the $k$th iteration of the for-loop.
If $k\geq \log \sqrt{n}$, the selective family in line 11 will select in some
round exactly one element from the set of local leaders. This follows
from the fact that the set of local leaders is not larger than $2^k$
in such case thus the $(N,2^k)$-selective family is sufficient to select
exactly one element from $L$.

Regarding time complexity, the $k$th iteration of the for-loop requires
$O(2^k\log N)$ rounds, which gives
$$O\left(\sum_{k=1}^{(\log n)/2} 2^k\log N\right)=O(\sqrt{n}\log N)$$
rounds. 
\comment{
Regarding time complexity, the $k$th iteration of the for-loop requires
$O(2^k\log(N/2^k)/k)$ rounds, which gives
$$O\left(\sum_{k=1}^{(\log n)/2} 2^k\log(N/2^k)/k\right)=O(\sqrt{n}\log N)$$
rounds. 
}
\end{proof}
\else
\fi

\comment{
\begin{corollary}
The leader election problem can be solved in time $O(\sqrt{n}\log(N))$.
\end{corollary}
\begin{proof}
\tj{BEDZIE WYNIKAC Z REDUKCJI I TU NIEPOTRZEBNE}
First, we solve the nontrivial move problem (Lemma~\ref{lem:ntmp}). Using a nontrivial round, the common
sense of direction can be established in $O(1)$ rounds. Finally, the leader can
be elected in $O(\log N)$ rounds by Lemma~\ref{lem:emptiness:leader}.
\end{proof}
}

\subsection{Position Discovery in the {\perceptive} model}\labell{sub:sec:disc:perc}

In this section we design an efficient solution for the position
discovery in the {\perceptive} model. Using results from the previous
section and Theorem~\ref{ABC},
we can assume that the leader is elected and the common sense of direction is established
in $O(\sqrt{n}\log N)$ rounds.
Throughout this section, we use a labeling of agents such that $a_1$ is the label of the leader and $a_i$ is the
label of the $i$th agent on the ring in the clockwise direction from the leader.

We solve the position discovery problem in two stages. First, each agent determines
its right ring distance to the leader (i.e., its {\em label}; note that a label denotes
the distance to the leader, not ID).  
In order to achieve this goal in the standard message passing model on a ring, linear
time is necessary. In order to perform this task faster, we use arithmetic relationships
between distances to collisions (\coll()) and  distances traversed in consecutive rounds (\pos()).
For appropriately designed protocol, an agent in ring distance $\leq d^2$ from the leader will
be able to learn its ring distance in $O(d\log N)$ rounds.
Then, using the knowledge about ring distances of agents to the leader, the position discovery
will be finally solved in the following way. Let $x_1,\ldots,x_n$ be the original distances between agents.
Here, we plan movements of agents in such a way that, for each agent and each round,
the distance to collision in the round and the distance traversed in the round gives a linear equation
over $x_1,\dots,x_n$ which is linearly independent from equations
derived before. In this way each round provides two new equations and $n/2$ rounds are sufficient
to determine the actual values of $x_1,\dots,x_n$,
since they give a system of $n$ independent linear equations over
$n$ variables.

\subsubsection{Ring distances}
Now, we design the RingDist protocol in which each agent learns
its right ring distance to the leader. 
Throughout this section, ring distance denotes the right ring distance from
the leader. We call it a {\em label} of an agent and denote the agent in ring
distance $i$ by $a_i$.

Let Shift($l$) for $l\in\NAT$ be a round in which
$\dir_{a_{i}}=\rright$ for each $i\in[l]$ and
$\dir_{a_{i}}=\lleft$ for $i>l$.
Moreover, Shift($-l$) is a round with directions of agents
opposite to their direction in Shift($l$).
Observe that the rotation index of Shift($l$) is equal to 
$$(l-(n-l))\ \modd n\equiv 2l\modd n.$$

RingDist works under assumption that (exactly) one distinguished
agent has the status leader (it is denoted $a_1$).
Each agent but the leader starts an execution of a protocol with unspecified
ring distance. The idea of Algorithm~\ref{alg:move:perceptive} is that the
agents gradually learn their ring distances in the following way:
\begin{itemize}
\item
The agents in ring distance $\leq 4$ learn their distances in step 1
(the same applies to the agents $\geq n-4$, although they learn merely their
relative values, without knowing $n$).
\item
In the $i$th iteration of the for-loop, the agents $a_k,a_{k+k},\ldots,a_{k+k^2}$ for $k=2^i$
learn their ring distances in the following way (see Fig.~\ref{fig:ring:dist}).
\begin{figure}[h]
\begin{center}
  \includegraphics[scale=0.45]{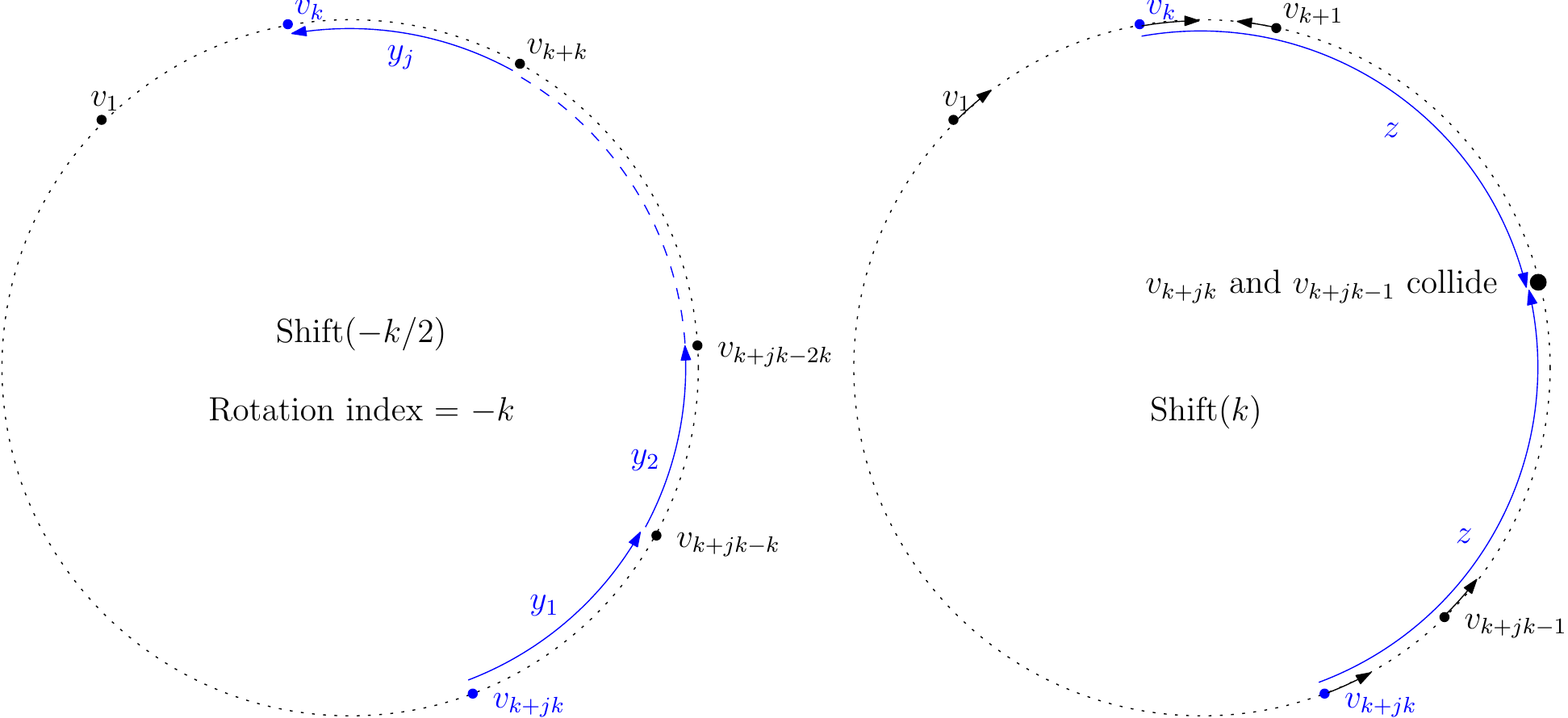}
  \caption{An illustration for Algorithm~\ref{alg:move:perceptive}. The agent's
	label is not $a_k+jk$ for any $j\leq k$ iff $2z$ is not equal to any of the sums
	$\sum_{i=1}^j y_i$.}
\label{fig:ring:dist}
\end{center}
\end{figure}
For each $l>k$, the value of $\coll()$ in Shift($k/2$) is equal to $z=(x_{l-k}+\cdots+x_{l-1})/2$
(see Prop.~\ref{prop:bounce} for $b_0=a_l$, $\dir=\lleft$ and thus
$b_i=a_{(l-i)\modd n}$).
On the other hand, if one applies Shift($-k/2$) several times,
the values of $\pos()$ in the $j$th executions of Shift($-k/2$) is equal to
$y_j=x_{l-jk}+\cdots+x_{l-(j-1)k+1}$, since the rotation index of Shift($-k/2$) is equal to $-k$.
Using these relationships, we see that there exists $j$ such that $2z=y_1+\cdots+y_j$
iff $l=k+k\cdot j$.
This observation is exploited in RingDist in order to determine ring distances of $a_1,\ldots a_{k+k^2}$
in the $i$th iteration of the main for-loop for $k=2^i$.

\item
The remaining agents $a_j$ for $j\leq k+k^2$ learn their distances
in the execution of 
line 8, as each agent knowing
its ring distance propagates it in the distance $2k$.
\end{itemize}
Then, it remains to guarantee that the for-loop is finished when all agents
know their ring distances and $2^i=O(\sqrt{n})$.
To this aim, we execute CheckCompleteness.
Note that the agent $a_n$ knows
that it is the last one already at the beginning (without knowing $n$), as it is
the left neighbour of the leader.
CheckCompleteness is a round in which all agents different from $a_n$ move
{\lleft}, while $a_n$ moves {\rright} iff it already knows its own
right ring distance (which in turn implies that every other agent knows its ring distance
as well).
Thus, the rotation index of this round is not zero iff each agent knows its ring distance.
\begin{algorithm}[]
	\caption{RingDist($a$)}
	\label{alg:move:perceptive}
	\begin{algorithmic}[1]
		\State \textbf{if} $a=a_1$: Diss(''leader'',4)\iffull\Comment{The leader $a_1$ broadcasts its message on ring distance $4$}\else\tj{\Comment{$a_1$ broadcasts on dist.\ $4$}}\fi
    \For{$i=1,2,\ldots,\log N$}
        \State $k\gets 2^i$
        \State For $j=1,\ldots,k$: Shift($-k/2$); $y_j\gets\pos()$
        \State Repeat $k$ times: Shift($k/2$) \iffull\Comment{Reverse the result of $k\times$ Shift($-k/2$)}\else\tj{\Comment{Reverse res. of l.\ 4}}\fi
        \State Shift($k$); $z\gets \coll()$; Shift($-k$)
        \State \textbf{if} $2z=y_1+\cdots+y_j$ for some $j$
				and $a\not\in\{a_1,\ldots,a_k\}$: 
				\State\phantom{abc}Set the ring distance of $a$ to $k+jk$; mark $a$\iffull\Comment{i.e., $a\gets a_{k+jk}$}\else\tj{\State\phantom{abc} (i.e., $a$ is $a_{k+jk}$ and $a$ is marked)}\fi
        \If{$a=a_{k+jk}$ for $j\leq k$ and $a$ marked}
					\State Diss($k+jk,k$)
					\iffull\Comment{Marked agents broadcast their ring dist.\ on distance $k$}\else\tj{ (i.e., marked agents \State broadcast their ring dist.\ on distance $k$)}\fi
				\EndIf
\State If CheckCompleteness: return\iffull\Comment{See description of the alg. for details}\fi
    \EndFor
    \end{algorithmic}
\end{algorithm}
In the following, we show more formally that the above described idea works. First, we make
an observation following from the definition of Shift (the rotation index of Shift($l$) is $2l$)
and Proposition~\ref{prop:bounce}.
\begin{proposition}\label{prop:shift}
Let $k=2^i$ for $i\leq \log N$.
Assume that agents $a_1,\ldots,a_k$ know their labels before the $i$th
iteration of the for-loop (and other agents know that they do not belong
to $\{a_1,\ldots,a_k\}$).
Then, for $l>k$ the values of $z$, $y_1,\ldots,y_k$ recorded by the agent $a_l$ satisfy the following
conditions in the iteration $i$ of the for-loop:
\begin{itemize}
\item $y_j=x_{l-kj}+x_{l-kj+1}+\cdots+x_{l-k(j-1)-1}$;
\item $z=(x_k+\cdots+x_{l-1})/2$.
\end{itemize}
\end{proposition}

\iffull
The following corollary is an immediate consequence of Proposition~\ref{prop:shift}.
\fi
\begin{corollary}\labell{cor:shift}
The condition $2z=y_1+\cdots+y_j$ is satisfied for an agent $a\not\in\{a_1,\ldots,a_k\}$ iff $a$ is in the right ring distance
$k+jk$ from the leader (i.e., $a=a_{k+jk}$).
\end{corollary}

\begin{lemma}\labell{lem:ring}
Assume that the leader is elected and all agents share common sense of direction.
Then, each agent $a$ determines its ring distance during the algorithm RingDist and the
algorithm lasts $O(\sqrt{n}\log N)$ rounds.
\end{lemma}
\iffull
\begin{proof}
Before the for-loop, the agents in ring distance $\leq 4$ are aware of their ring distance,
while the agent $a_n$ knows that it is ``the last one''. 
We show by induction that, after the $i$th iteration of the main for-loop, the agents
$a_1,\ldots,a_{k^2}$ know their labels for $k=2^i$.
As the base step is obvious, assume inductively that before the $i$th iteration,
the agents $a_1,\ldots, a_{(2^{i-1})^2}$ know their labels. Thus, in particular
$a_1,\ldots,a_{k}$ know their labels for $k=2^i$ since $2^i<(2^{i-1})^2$.
This assures that agents are able to perform all steps in the $i$ iteration
of the main for-loop.
Then, Proposition~\ref{prop:shift}
and Corollary~\ref{cor:shift} imply that each agent
$a_{k+jk}$ for $j\in[k]$ and $k=2^i$ becomes aware of its ring distance
before 
dissemination of distances in line 10.
In line 10 agents $a_k,a_{2k},\ldots,a_{(k-1)k},a_{k^2}$
\tj{broadcast} information about their ring distances to agents in their ring distance
$\leq k$.
Thus, the remaining agents $a_l$ for $l\leq k+k^2+k$ learn their ring distances
from the agents $a_{k+k},\ldots,a_{k+2k},\ldots,a_{k+k^2}$.
Finally, for the smallest $i$ such that $2^i+(2\cdot 2^i)^2+2^i\geq n$, the for-loop
is finished and all agents know their ring distances.
In order to execute the $i$th iteration, $O(2^i\log N)$ rounds are sufficient
(by Corollary~\ref{cor:sdisseminate}, dissemination of distances by marked agents
can be done in $O(2^i+\log N)=O(2^i\log N)$ rounds).
Time complexity of the whole procedure is
$$O(\sum_{i=1}^{\frac12\log n} 2^i\log N)=O(\sqrt{n}\log N).$$
\end{proof}
\else
\fi

\subsubsection{Position discovery}
In this section we describe a solution for the position discovery
problem based on protocols presented before.
Recall that, given  the common sense of direction and the leader, one can obtain a round
with rotation index $2$ by assigning the direction $\lleft$ to all agents but the leader.
If $n$ is odd, this gives a solution to the position discovery problem in $n$ rounds.
The goal of this section is to get advantage of information provided by positions of 
the first collision in a round, in order to decrease time from $n$ to $n/2$ and manage
the case that $n$ is even.

Here, we assume that the leader $a_1$ is elected, the agent(s) in the right ring distance
$i$ from the leader is $a_{i+1}$ and each agent $a_i$ knows $i$ (see Corollary~\ref{cor:shift}). Moreover, we assume that
$n$ is even (one can easily build a similar solution for odd $n$).
Let $x_i$ denote the distance on the ring between the agent $a_i$ and the agent $a_{i+1}$
(or $a_1$ if $i=n$). (Note that this is the geometric distance on the ring, not the ring distance!)

Let Convolution($j$) be a round in which the agents' directions are as 
\iffull
follows (see Figure~\ref{fig:red:strong}):
\else
follows:
\fi
\begin{itemize}
\item[]
$\dir_{a_{2i-1}}=\rright$, $\dir_{a_{2i}}=\lleft$ for each $i\in[n/2]$, with an
exception: $\dir_{a_{2j}}=\rright$.
\end{itemize}
\iffull
\begin{figure}[H]
\begin{center}
  \includegraphics[scale=0.45]{conv.pdf}
  \caption{An illustration of Convolution$(j)$ and Pivot$(n)$. In Convolution($j$), the blue agents are the only one which do not collide with a neighbor in the middle of the distance between them. In Pivot$(j)$, blue agents are the only one which collide with a neighbor in the middle of distance between them.}
\label{fig:red:strong}
\end{center}
\end{figure}
\else
\fi

Let Pivot($j$) be a round in which the agents' directions are as follows:
$\dir_{a_{j+1}}=a_{j+2}=\cdots=a_{j+n/2}=\lleft$ and
$\dir_{a_{j}}=\dir_{a_{j-1}}=\cdots=\dir_{a_{j-n/2+1}}=\lleft$.
(Here, the subscript indices are calculated modulo $n$ and $a_0$ is identified
with $a_n$.)
Observe that the rotation index of Convolution($i$) is equal to $2$ and the rotation
index of Pivot($i$) is equal to $0$ for each $i$.
\begin{algorithm}[]
	\caption{Distances($a$)}
	\label{alg:disc:perc}
	\begin{algorithmic}[1]
    \For{$i=1,2,\ldots,n/2$}
        \State Convolution$(\frac{n-2(i-1)}{2})$
    \EndFor
    \State Pivot($n$); Pivot($n-1$); Pivot($n-2$);
    \end{algorithmic}
\end{algorithm}
In the following, we show that information collected during an
execution of Algorithm~\ref{alg:disc:perc} can determine original
positions of all other agents. 
\begin{proposition}\label{prop:linear}
After the for-loop of Algorithm~\ref{alg:disc:perc}, the following
conditions hold:
\begin{enumerate}
\item[(a)]
the agent $a_{2i-1}$ 
can
determine the values of
$x_{1}, x_2, \ldots,x_{n-2}$ and $x_{n-1}+x_n$.
\item[(b)]
the agent $a_{2i}$ 
can determine the values of
$x_1,x_2,\ldots,x_{n-3}$, $x_n$, and $x_{n-2}+x_{n-1}$
\end{enumerate}
for each $i\in[n/2]$.
\end{proposition}
\iffull
\begin{proof}
First, observe that
the agent $a_k$ stays at the original position
of $a_{k+2(i-1)}$ before the $i$th execution of line 2, since the rotation index
of Convolution is equal to $2$. Therefore, the only agent with ``swapped'' direction
stays in each round of the for loop at the original position of $a_{j}$
for $j=\frac{n-2(i-1)}{2}+2(i-1)=\frac{n}2$.
As a result, the agent $a_k$ learns in the $(i+1)$st execution of line 2 exactly the same values as $a_{k+2i}$
in the first execution of line 2.
Consider this first round. As $a_{2j-1}$ and $a_{2j}$ start
the round moving towards each other for each $j\in[n/2-1]$ (i.e., $j< n/2$),
they collide after traversing the distance $x_{2j-1}/2$ and therefore learn $\coll()=x_{2j-1}/2$.
On the other hand, the rotation index of Convolution with any
parameter is $2$ and thus each agent $a_k$ learns $\pos()=x_k+x_{k+1}$.
All these observations lead to the following conclusion:
\begin{itemize}
\item
Each agent with odd label 
learns the values of
$x_{2j-1}+x_{2j}$ for each $j\in[n/2]$ (registered as $\pos()$ in consecutive rounds)
and $x_{2k-1}$ (corresponding to $\coll()$) for each $k\in[1,n/2-1]$.
\item
Each even agent 
learns the values of
$x_{2(j-1)}+x_{2j-1}$ for each $j\in[n/2]$
and $x_{2k-1}$ for each $k\in[1,n/2-1]$.
\end{itemize}
Given, the above, each agent can deduce the values enumerated
in the proposition.
E.g., each agent with add label solves the 
system of equations $x_{2j-1}+x_{2j}=d_j$ for $j\leq n/2$
and $x_{2j-1}=c_j$ for $j<n/2$, where $d_j$ and $c_j$ are
the values observed as $\pos()$ and $\coll()$ in various rounds.
\end{proof}
\else
\tj{The above proposition uses the fact that each execution of 
Convolution gives information about the sum of two consecutive
$x_i$'s (as the distance between an agent's position at the beginning and the end of a round) and about a particular $x_j$ (the distance to the first collision is equal to halve of $x_j$ for some $j$).}
\fi
Now, observe that
\begin{itemize}
\item
The agent $a_{2i-1}$ has the first collision during Pivot($n$) in distance
$x_n/2+(x_1+\cdots+x_{2i-2})/2$ for each $i\in[n/4]$.
As $a_{2i-1}$ knows $x_1,\ldots,x_{2i-2}$ and $x_{n-1}+x_n$ by Proposition~\ref{prop:linear}(a), it can determine $x_n/2$ from Pivot$(n)$ and 
therefore also $x_{n-1}$.
\item
The agent $a_{2i-1}$ has the first collision during Pivot($n-1$) in distance
$x_{n-1}/2+(x_{2i-1}+\cdots+x_{n-2})/2$ for each $i\in[n/4+1,n/2-1]$.
As it knows $x_{2i-1},\ldots,x_{n-2}$ and $x_{n-1}+x_n$, it can
determine $x_{n-1}$ from Pivot$(n-1)$ and then $x_n$ as well.

A similar reasoning works for $a_{n-1}$ as well.
\end{itemize}
By combining the above with Prop.~\ref{prop:linear}(a), one can conclude that each 
agent $a_i$ with odd label $i$ knows original positions of all agents. A similar argument
applies for even agents and executions of Pivot($n-1$) and Pivot($n-2$),
since Prop.~\ref{prop:linear}(b) can be seen as Prop.~\ref{prop:linear}(a) ``shifted'' by $-1$.

\iffull
Thus, we obtain the following conclusion.
\else
\fi
\begin{lemma}\labell{lem:distances}
The protocol Distances (Alg.~\ref{alg:disc:perc}) solves the position discovery problem, provided
the leader ($a_1$) is elected, agents share common sense of direction
and each agent knows its ring distance.
\end{lemma}

\iffull
By combining the subroutines described before, we get the following result.
\fi
\begin{theorem}
The position discovery problem can be solved in the perceptive model
in $n/2+O(\sqrt{n}\log N)$ rounds.
\end{theorem}

\iffull
\section{On distinguishing parity of $n$ and small values of $n$}\labell{sec:special}
\input{distinguish.tex}
\else
\fi

\iffull
\section{Randomized Algorithms}\labell{sec:random}
\input{random.tex}
\else
\fi

\iffull
\section{Conclusions and open problems}\labell{sec:open}
In this paper, we evaluated complexity of position discovery
and some coordination problems in synchronous geometric
ring networks in considered settings.
\iffull
In the {\basic} and {\lazy} model, we have shown an exponential gap in complexity 
of coordination problems between the cases that $n$ is odd
and even.
Moreover, we established a relationship of these problems with
a purely combinatorial structure called a distinguisher.
In all cases in which the location discovery problem is feasible,
we provided efficient, (almost) optimal solutions.
\fi
An interesting open problem is to give explicit constructions
for the nontrivial move problem for even $n$ in the {\basic}
and {\lazy} model. (In this paper, their existence
is merely shown by the probabilistic method.)
The exact complexity of the nontrivial move problem in the {\perceptive}
model is also not known; we have only provided the upper bound
$O(\sqrt{n}\log N)$.
\iffull
We also do not know whether the additive term $\frac{n}{c}$
in solutions for the position discovery for anonymous networks
is necessary  (see Theorems~\ref{th:random:pd}, \ref{th:random:coord}). 
\fi
\fi

\iffull
\else
\fi
\section*{Acknowledgements}

This work was supported by Polish National Science Centre grant
DEC-2012/06/M/ST6/00459. 


\iffull
\bibliographystyle{IEEEtran}
\bibliography{references}

\begin{thebibliography}{10}
\providecommand{\url}[1]{#1}
\csname url@samestyle\endcsname
\providecommand{\newblock}{\relax}
\providecommand{\bibinfo}[2]{#2}
\providecommand{\BIBentrySTDinterwordspacing}{\spaceskip=0pt\relax}
\providecommand{\BIBentryALTinterwordstretchfactor}{4}
\providecommand{\BIBentryALTinterwordspacing}{\spaceskip=\fontdimen2\font plus
\BIBentryALTinterwordstretchfactor\fontdimen3\font minus
  \fontdimen4\font\relax}
\providecommand{\BIBforeignlanguage}[2]{{%
\expandafter\ifx\csname l@#1\endcsname\relax
\typeout{** WARNING: IEEEtran.bst: No hyphenation pattern has been}%
\typeout{** loaded for the language `#1'. Using the pattern for}%
\typeout{** the default language instead.}%
\else
\language=\csname l@#1\endcsname
\fi
#2}}
\providecommand{\BIBdecl}{\relax}
\BIBdecl

\bibitem{ASY}
H.~Ando, I.~Suzuki, and M.~Yamashita, ``Formation and agreement problems for
  synchronous mobile robots with limited visibility,'' in \emph{Intelligent
  Control}, 1995, 
  pp. 453--460.

\bibitem{ArkinFM00}
\BIBentryALTinterwordspacing
E.~M. Arkin, S.~P. Fekete, and J.~S.~B. Mitchell, ``Approximation algorithms
  for lawn mowing and milling,'' \emph{Comput. Geom.}, vol.~17, no. 1-2, pp.
  25--50, 2000. 
\BIBentrySTDinterwordspacing

\bibitem{AW-book}
H.~Attiya and J.~Welch, \emph{Distributed Computing}.\hskip 1em plus 0.5em
  minus 0.4em\relax McGraw-Hill, 1998.

\bibitem{BenderS94}
M.~A. Bender and D.~K. Slonim, ``The power of team exploration: Two robots can
  learn unlabeled directed graphs,'' in \emph{FOCS}.\hskip 1em plus 0.5em minus
  0.4em\relax IEEE Computer Society, 1994, pp. 75--85.

\bibitem{ChalopinFMS10}
\BIBentryALTinterwordspacing
J.~Chalopin, P.~Flocchini, B.~Mans, and N.~Santoro, ``Network exploration by
  silent and oblivious robots,'' in 
	{WG} 2010, 
	LNCS 6410, 2010, pp. 208--219. 
\BIBentrySTDinterwordspacing

\bibitem{CieliebakFPS12}
M.~Cieliebak, P.~Flocchini, G.~Prencipe, and N.~Santoro, ``Distributed
  computing by mobile robots: Gathering,'' \emph{{SIAM} J. Comput.}, vol.~41,
  no.~4, pp. 829--879, 2012.

\bibitem{ClementiMS03}
A.~E.~F. Clementi, A.~Monti, and R.~Silvestri, ``Distributed broadcast in radio
  networks of unknown topology,'' \emph{Theor. Comput. Sci.}, vol. 302, no.
  1-3, pp. 337--364, 2003.

\bibitem{Cooley05}
B.~Cooley and P.K.~Newton, ``Iterated impact dynamics of $N$-beads on a ring,''
\emph{SIAM Review}, vol.~47, no.~2, pp. 273--300, 2005.

\bibitem{CzyzowiczDKP14}
\BIBentryALTinterwordspacing
J.~Czyzowicz, S.~Dobrev, E.~Kranakis, and E.~Pacheco, ``Survivability of swarms
  of bouncing robots,'' in \emph{LATIN} 2014, 
	LNCS 8392, 
	pp.
  622--633. 
\BIBentrySTDinterwordspacing

\bibitem{CzyzowiczGKKPP12}
J.~Czyzowicz, L.~G\k{a}sieniec, A.~Kosowski, E.~Kranakis, O.~M. Ponce, and
  E.~Pacheco, ``Position discovery for a system of bouncing robots,'' 
	{DISC} 2012,
	LNCS 7611. 
	Springer, 2012, pp. 341--355.

\bibitem{CzyzowiczKP13}
\BIBentryALTinterwordspacing
J.~Czyzowicz, E.~Kranakis, and E.~Pacheco, ``Localization for a system of
  colliding robots,'' in 
	{ICALP} 2013, 
	LNCS 7966, 
	Springer, 2013, pp. 508--519. 
\BIBentrySTDinterwordspacing

\bibitem{CzyzowiczLP11}
\BIBentryALTinterwordspacing
J.~Czyzowicz, A.~Labourel, and A.~Pelc, ``Optimality and competitiveness of
  exploring polygons by mobile robots,'' \emph{Inf. Comput.}, vol. 209, no.~1,
  pp. 74--88, 2011. 
\BIBentrySTDinterwordspacing

\bibitem{DasFSY10}
S.~Das, P.~Flocchini, N.~Santoro, and M.~Yamashita, ``On the computational
  power of oblivious robots: forming a series of geometric patterns,'' in
	{PODC} 2010, 
  {ACM}, 2010, pp. 267--276.

\bibitem{DengKP98}
\BIBentryALTinterwordspacing
X.~Deng, T.~Kameda, and C.~H. Papadimitriou, ``How to learn an unknown
  environment {I:} the rectilinear case,'' \emph{J. {ACM}}, vol.~45, no.~2, pp.
  215--245, 1998. 
\BIBentrySTDinterwordspacing

\bibitem{FlocchiniPSW08}
P.~Flocchini, G.~Prencipe, N.~Santoro, and P.~Widmayer, ``Arbitrary pattern
  formation by asynchronous, anonymous, oblivious robots,'' \emph{Theor.
  Comput. Sci.}, vol. 407, no. 1-3, pp. 412--447, 2008.

\bibitem{FraigniaudGKP06}
\BIBentryALTinterwordspacing
P.~Fraigniaud, L.~G\k{a}sieniec, D.~R. Kowalski, and A.~Pelc, ``Collective tree
  exploration,'' \emph{Networks}, vol.~48, no.~3, pp. 166--177, 2006. 
\BIBentrySTDinterwordspacing

\bibitem{FranklF85}
P.~Frankl and Z.~F{\"u}redi, ``Forbidding just one intersection,'' \emph{J.
  Comb. Theory, Ser. A}, vol.~39, no.~2, pp. 160--176, 1985.

\bibitem{FriedetzkyGGM12}
T.~Friedetzky, L.~G\k{a}sieniec, T.~Gorry, and R.~Martin, ``Observe and remain
  silent (communication-less agent location discovery),'' in \emph{MFCS}, 
	LNCS 7464, 
	Springer, 2012, pp.
  407--418.

\bibitem{GrebinskiK00}
V.~Grebinski and G.~Kucherov, ``Optimal reconstruction of graphs under the
  additive model,'' \emph{Algorithmica}, vol.~28, no.~1, pp. 104--124, 2000.

\bibitem{HoffmannIKK01}
\BIBentryALTinterwordspacing
F.~Hoffmann, C.~Icking, R.~Klein, and K.~Kriegel, ``The polygon exploration
  problem,'' \emph{{SIAM} J. Comput.}, vol.~31, no.~2, pp. 577--600, 2001.
\BIBentrySTDinterwordspacing

\bibitem{kkm-book}
E.~Kranakis, D.~Krizanc, and E.~Markou, \emph{The Mobile Agent Rendezvous
  Problem in the Ring}.\hskip 1em plus 0.5em minus 0.4em\relax Morgan and
  Claypool Publishers, 2010.

\bibitem{lynch-book}
N.~Lynch, \emph{Distributed Algorithms}.\hskip 1em plus 0.5em minus 0.4em\relax
  Morgan Kaufmann Publishers, 1996.

\bibitem{Susca07}
S.~Susca, F.~Bullo, and S.~Mart\'{i}nez, ``Synchronization of beads on a ring,''
\emph{CDC} 2007, pp. 4845--4850.

\bibitem{SuzukiY99}
I.~Suzuki and M.~Yamashita, ``Distributed anonymous mobile robots: Formation of
  geometric patterns,'' \emph{{SIAM} J. Comput.}, vol.~28, no.~4, pp.
  1347--1363, 1999.


\end{thebibliography}
\else

\fi

\newpage
\section*{Appendix: Proofs omitted from Section~\ref{sec:basic}}

\noindent {\em Proof of Proposition~\ref{prop:reduction}}

First, assume that $n$ is known and $\m{A}$ solves the weak nontrivial move problem.
Observe that, until the first round of $\m{A}$ with a weak nontrivial move, the 
only information available
to each agent is that its starting position in a round is equal to its position at the end
of a round. Thus, its behavior can be defined by a sequence of sets $S_1, S_2, \ldots$,
such that the agent $a$ chooses direction {\rright} in round $i$ (provided no nontrivial
move appeared before) if and only if $a\in S_i$.
Let us fix which sense of direction is ``correct''.
Then, consider the situation in which the set of agents $X_1$ with the 
correct sense of direction and
the set of agents $X_2$ with the incorrect sense of direction satisfy $|X_1|=|X_2|=n/2$.
Let $m_1=|X_1\cap S_i|$, $m_2=|X_2\cap S_i|$.
Then,
the rotation index (mod $n$) in round $i$ is
$$\begin{array}{lll}
& & (|X_1\cap S_i|+|X_2\setminus S_i|) - (|X_1\setminus S_i|+|X_2\cap S_i|)  \\
                                       & & =  (m_1+n/2-m_2)-(n/2-m_1+m_2)\\
& & =  2(m_1-m_2).
\end{array}$$
And therefore the $i$th round of $\m{A}$ gives a (weak) nontrivial move if and only 
if $2(m_1-m_2)\not\in\{0,n\}$, 
which implies $m_1\neq m_2$. On the other hand, $m_1\neq m_2$ is equivalent to the fact that
$S_i$ distinguishes $X_1$ and $X_2$.
In conclusion, the sequence $S_1,S_2,\ldots$ defining $\m{A}$ is a $(N,n/2)$-distinguisher.

For unknown $n$, the result follows from the above reasoning and the fact that $\m{A}$
has to tackle arbitrary even $n\leq N$ which reflects the difference between
a standard $(N,n)$-distinguisher and its strong counterpart.
\IEEEQEDclosed
%

Before giving the proof of Lemma~\ref{lem:weak}, we provide a lower 
bound on the size of a strong $(N,n)$-distinguisher with
a simple proof based on a counting argument
(a similar bound in another context was given e.g.\ in \cite{GrebinskiK00}).
Although this result is subsumed by Lemma~\ref{lem:weak}, we provide it 
to give some intuition before a more complicated, and less intuitive proof,
of Lemma~\ref{lem:weak}.
\begin{lemma}\labell{lem:strong}
If $S$ is a strong $(N,f)$-distinguisher for any $N>4$ 
and $f\,:\,\NAT\times\NAT\to\NAT$, then 
$f(N,n)=\Omega\left(\frac{n\log(N/n)}{\log n}\right)$.
\end{lemma}

\begin{proof}
First, we show  that a strong $(N,f)$-distinguisher  $\mathcal{S}$ 
satisfies the property that for each two different sets $X_1,X_2\subset [N]$ such that
$|X_1|=|X_2|=n$, there exists $i\leq f(N,n)$ such that
$|X_1\cap S_i|\neq |X_2\cap S_i|$
(note that $X_1$ and $X_2$ do not have to be disjoint!).
Indeed, assume to the contrary that this is not the case 
for $\mathcal{S}$, and thus $|X_1\cap S_i|=|X_2\cap S_i|$
for some different sets $X_1,X_2$ of size $n>1$ and each $i\in[f(N,n)]$.
Let $Y_1=X_1\setminus X_2$ and $Y_2=X_2\setminus X_1$. Then, $Y_1\cap Y_2=\emptyset$,
$|Y_1|=|Y_2|\leq n$ and $|Y_1\cap S_i|=|Y_2\cap S_i|$ for each $i\in[f(N,n)]$. This implies
that $\mathcal{S}$ is not a strong $(N,f)$-distinguisher, which is a contradiction.

Let $\mathcal{S}=(S_1,\ldots,S_k)$ be a strong $(N,f)$-distinguisher. The above observation implies
that, for any $X\neq X'$, $X,X'\subset[N]$ of size $n$, the sequences $|X\cap S_1|,\ldots,|X\cap S_k|$
and $|X'\cap S_1|,\ldots,|X'\cap S_k|$ are not equal, where
$k=f(N,n)$.
As each $S_i$ gives at most $n+1$ possible values of $|X\cap S_i|$ for $X\subset[N]$ of size $n$,
and there are $N\choose n$ subsets of $[N]$ of size $n$, we obtain the following bound
$$k\geq \log_{n+1}{N\choose n}=\Omega\left(\frac{\log{N\choose n}}{\log(n+1)}\right)=\Omega\left(\frac{n\log(N/n)}{\log n}\right)$$
for $n>1$.
\end{proof}
It turns out that the result of Lemma~\ref{lem:strong} can be strengthened,
to give Lemma~\ref{lem:weak}.  
However, our proof of this fact is much more complicated.  It applies
techniques from~\cite{ClementiMS03}, designed for proving lower bounds on size of
selective families.
We stress here that the lower bound for a strong variant of a distinguisher does not imply an analogous lower bound for
a ``standard'' variant of a distinguisher.
As observed in the proof of Lemma~\ref{lem:strong}, the
prefix of size $f(N,n)$ of a stong $(N,f)$-distinguisher
gives an opportunity to ``distinguish'' \textbf{each} pair of sets
of size $n$. On the other hand, a standard $(N,n)$-distinguisher
is supposed to give a difference only on \textbf{disjoint} sets
of size $n$.

\noindent {\em Proof of Lemma~\ref{lem:weak}}

Let us first stress that the calculations from the previous lemma do not apply here, since
a (``standard'') distinguisher does not have to ``distinguish'' small sets, so it does not have to distinguish
non-disjoint sets of size $n$ either.

Let $G(V,E)$ be a graph, whose vertices are all $2n$-subsets of $[N]$, where the 
edges connect vertices
corresponding to sets which have exactly $n$ common elements.
That is, $(X_1,X_2)\in E$ for $X_1,X_2\in V$ if and only if $|X_1\cap X_2|=n$.
Let $\alpha(G)$ and $\chi(G)$ denote the size of the largest independent set of $G$
and the chromatic number of $G$, respectively.
We claim that
\begin{eqnarray}
\log \chi(G) & \geq & \frac16 n\log(N/(2n)) \ \ \ \textrm {and}\label{e:color1}  \\
\log \chi(G) & \leq  & |\m{S}|\log(2n+1). \label{e:color2}
\end{eqnarray}
Proof of (\ref{e:color1}):\newline
We use the fact that $\chi(G)\geq \frac{|V|}{\alpha(G)}$.
Moreover, as each independent set of $G$ is a $(N,2n,n)$-intersection free family
of sets,
Fact~\ref{fac:free} implies that
$$\log\alpha(G)\leq \frac{22}{12}n\log(N/(2n)).$$
Therefore
$$\begin{array}{rclcl}
\log \chi(G) & \geq & \log |V|-\log\alpha(G) \\
                & \geq & \log {N\choose 2n}-\frac{22}{12}n\log(N/(2n)) \ \ \ \textrm {and} \\
               & \geq & 2n\log(N/(2n))-\frac{22}{12}n\log(N/(2n)) \\
               & = & \frac{1}{6}n\log(N/(2n)), 
\end{array}
$$
which gives (\ref{e:color1}). In the third inequality, we use the relation 
${a \choose b}\ge \left(\frac{a}{b}\right)^b$.

\noindent Proof of (\ref{e:color2}):\newline
\comment{
Let $\m{S}=(S_1,\ldots,S_m)$ be a $(N,n)$-distinguisher, and let $G_i(V_i,E_i)$ be a graph associated with $S_i$ such that
\begin{enumerate}
\item
$V_i=V$ and
\item
$(X_1,X_2)\in E_i$ if and only if $(X_1,X_2)\in E$ and
$|(X_1\setminus X_2)\cap S_i|\neq |(X_2\setminus X_1)\cap S_i|$
(i.e.\ $S_i$ distinguishes $X_1\setminus X_2$ and $X_2\setminus X_1$).
\end{enumerate}
Now, observe that  the set $\{X\,|\ |S_i\cap X|=p\}$ is independent in $G_i$
for each $p\in[0,n]$. 
Therefore, $\chi(G_i)\leq n+1$ since
the graphs $G_1,\ldots,G_m$ share the same set of vertices, so
$\chi(\bigcup_{i=1}^m G_i)\leq \Pi_{i=1}^m \chi(G_i)$.
The assumption that $\m{S}$ is a $(N,n)$-distinguisher implies that $\bigcup_i E_i=E$.
Therefore,
$$
\begin{array}{rcl}
\log \chi(G)  & \leq & \log \chi(\bigcup_i G_i) \\
                 & \leq & \log\left\{ \Pi_{i=1}^m \chi(G_i)\right\}  \\
                 & =   &\sum_{i=1}^m\log\chi(G_i)\leq m\log(n+1),
\end{array}
$$
which proves (\ref{e:color2}).
}
Let $\m{S}=(S_1,\ldots,S_m)$ be a $(N,n)$-distinguisher.
Observe that for any two sets $X_1,X_2$ such that $|X_1\cap X_2|=n$ there exists $S_i$
such that $|S_i\cap X_1|\neq |S_i\cap X_2|$.
In other words for any tuple $(p_1,\ldots,p_m),p_i\in[0,2n]$
the set $\{X\,:\forall i \ |S_i\cap X|=p_i\}$ is independent in $G$.
Therefore, $\chi(G)\leq (2n+1)^m$.
Thus
$$
\log \chi(G)  \leq m\log(2n+1),
$$
which proves (\ref{e:color2}).

Finally, observe that (\ref{e:color1}) and (\ref{e:color2}) imply the statement
of the lemma.
\IEEEQEDclosed

\noindent {\em Proof of Corollary~\ref{cor:lower:move}}

The result 
follows directly from Proposition~\ref{prop:reduction}
and Lemma~\ref{lem:weak}.
\IEEEQEDclosed

\noindent {\em Proof of Theorem~\ref{thm:selector}}

Let us choose a sequence $\m{S}$ of sets $S_1,S_2,\ldots$ probabilistically, such that each
$x\in[N]$ belongs to $S_i$ with probability $1/2$, where all choices are independent.
Then, our algorithm is defined such that, in round $i$, the agents with IDs in $S_i$
choose direction {\rright} and the other ones choose the direction {\lleft}.
We show that the family $\m{S}=(S_1,\ldots,S_k)$ chosen in this way gives a protocol
solving the nontrivial move problem with positive probability, provided the size
$n$ of the network is smaller than $N/3$. 
That is, the following event holds with positive probability:
for each $X\subset [N]$ such that $|X|<N/3$, the nontrivial move appears during
an execution of the prefix of $\m{S}$ of size $O(n\log(N/n)/\log n)$, where $n=|X|$.
Then we build a sequence $\m{C}$ of size $O(N/\log N)$ which gives a nontrivial move
on each $X\subset [N]$ of size at least $N/3$.
Thus, by interleaving $\m{S}$ and $\m{C}$,  the theorem holds thanks to the probabilistic
method.

Let us fix a set of IDs $A\subset[N]$ of size $n$ and assign sense of directions to them
such that $A=A_c\cup A_i$, where $A_c$ is the set of agents with correct sense
of directions, $|A_c|=n_c$ and $|A_i|=n-n_c$.
Recall that a round does {\em not} give a nontrivial move if and only if it is a
$(0,n)$-round, $(n,0)$-round, $(n/2,n/2)$-round, $(3n/4,n/4)$-round,
or a $(n/4,3n/4)$-round. Then, 
for a round defined by $S_i$ as above, we have:
$$\begin{array}{rcll}
\text{Prob}((n/2,n/2)\text{-round})&=&\frac1{2^n}\sum_{j=0}^{\min(n_c,n/2)} {n_c\choose j}{n-n_c\choose n/2-j} \\
                                               &=&\frac1{2^n}{n\choose {n/2}}\leq \frac{c_0}{n^{1/2}}, \\ 
\text{Prob}((0,n)\text{-round})&=&\frac1{2^n}{n_c\choose 0}{n-n_c\choose n-n_c} =\frac1{2^n}, \\
\text{Prob}((n,0)\text{-round})&=&\frac1{2^n}{n_c\choose n_c}{n-n_c\choose 0} =\frac1{2^n}, \\
\text{Prob}((n/4,3n/4)\text{-round})&=&\frac1{2^n}\sum_{j=0}^{\min(n_c,3n/4)} {n_c\choose j}{n-n_c\choose 3n/4-j} \\
                                                 &=&\frac1{2^n}{n\choose n/4}=1/2^{\Theta(n)}, \ \textrm{and} \\
\text{Prob}((3n/4,n/4)\text{-round})&=&\text{Prob}((n/4,3n/4)\text{-round})  \\
                                                 &=&1/2^{\Theta(n)}.\\
\end{array}$$
In the above calculations, we use the relationship that
$\sum_{i=0}^{\min(a,c)} {a\choose i} {b-a\choose c-i}={b\choose c}$ and Stirling's formula 
which determines the constant $c_0$ in the first row.
The above estimations imply that the probability that a round defined by $S_i$
is a trivial move for $|A|=n$ is at most $c_1/\sqrt{n}$ for some 
constant $c_1$, provided $n$ is large enough.
Let us consider all sets of IDs $A$ such that $|A|\in[2^{i-1},2^i)$, 
for $i$ such that $2^i<N/3$.
Let $k=c\frac{2\log{N\choose 2^i}}{i-1}$ for a large enough constant $c$ whose value
will be determined later.
By $E_i$ we denote the event that a sequence of sets $S_1,\ldots,S_k$ does {\em not}
give a nontrivial move for
all sets $A$ whose size is in $[2^{i-1},2^i)$. Then, 
$$\begin{array}{rclcl}
\text{Prob}(E_i) &\leq& \sum_{d=2^{i-1}}^{2^i}(\text{Prob}(\text{triv.\ move on a set of size }d))^k\cdot{N\choose d} 2^d \\
&\leq&
\sum_{d=2^{i-1}}^{2^i}\frac{{N\choose d} 2^d\cdot c_1}{2^{(i-1)k/2}}
\leq c_1\sum_{d=2^{i-1}}^{2^i}\frac{{N\choose d}^2}{{N\choose 2^i}^3}\\
&\leq& c_1\sum_{d=2^{i-1}}^{2^i}\frac{1}{{N\choose 2^i}}\leq
c_1\sum_{d=2^{i-1}}^{2^i}\frac{1}{2^i}<c_1\frac1{2^{i-2}}.
\end{array}$$
%
In the above calculations, we use the following facts:
\begin{itemize}
\item
${N \choose d}2^d$ is the number of possible choices of sets of size $d$, and senses of direction of elements of these sets (used in the first inequality);
\item
$\text{Prob}(\text{triv.\ move on a set of size }d)\leq \frac{c_1}{\sqrt{d}}\leq \frac{c_1}{2^{(i-1)/2}}$ 
(used in the second inequality);
\item
$2^{(i-1)k/2}\geq {N\choose 2^i}^{c}\ge {N\choose 2^i}^{3}$ for $c\ge 3$ (which 
follows from the fact that $k=c\frac{2\log{N\choose 2^i}}{i-1}$; used in the third
inequality);
\item
$2^d\leq {N\choose d}$
for $d\leq N/3$ (used in the third inequality);
\item
${N\choose d}\leq {N\choose 2^i}$  
for $d\leq N/3$ (used in the fourth inequality);
\item
${N\choose 2^i}\ge (N/2^i)^{2^i}\ge 2^i$ if $2^i<N/2$ (used in the fifth inequality).
\end{itemize}
Let $i_0=\lceil \log 4c_1\rceil+1$ and $i_1=\lfloor \log (N/3)\rfloor$.
The above calculations show that, the union of events $E_{i_0}, E_{i_0+1},\ldots,E_{i_1}$
holds with probability $\sum_{i_0}^{i_1}c_1/2^i<1/2$ for $c>3$.
Therefore, by the probabilistic method, the sequence $\m{S}$ gives a nontrivial move for each set of IDs of size
in $[2^{i_0},2^{i_1}]=[4c_1,N/c]$. It remains to tackle the cases that $n<2^{i_0}$ and $n>2^{i_1}$.

As for $n<2^{i_0}$, note that $2^{i_0}$ is a constant independent of $n$.
Thus the number of sets of size $<2^{i_0}$ is polynomial wrt $N$, while the
probability that a round gives a nontrivial move for a given set is larger
than some positive constant independent of $N$. Therefore on a sufficiently
long prefix of $\m{S}$ of length $O(\log N)=O(n\log(N/n)/\log n)$, the nontrivial
move appears with for each set of size
$<2^{i_0}$ with probability $1-1/N$.

Now, we consider the case that the size $n>2^{i_1}>N/3$. 
The number of such sets is upper bounded by $2^N$.
And, for each such set, each round gives a nontrivial move with probability at least
$c'/\sqrt{N}$ for a constant $c'$. By a simple calculation, one can show that
the nontrivial move appears for each such set on a long enough prefix
of $\m{S}$ of size $O(N/\log N)$ with probability $1-1/N$.
More precisely, on a prefix of $c''\log N$, the probability that there is a
set without a nontrivial move is smaller than
$$2^N (c'/\sqrt{N})^{c''\log N}<1/N$$
for a large enough constant $c''$.
\IEEEQEDclosed

\end{document}